\documentclass[12pt, draftclsnofoot, onecolumn]{IEEEtran}
\usepackage[utf8]{inputenc}
\usepackage[english]{babel}
\usepackage{cite}
\usepackage{setspace}
\usepackage{graphicx}
\usepackage{fancyhdr}
\usepackage{amsmath}
\usepackage{algorithm}
\usepackage{algcompatible}
\usepackage{array}
\usepackage{amsthm}
\usepackage[caption=false,font=normalsize,labelfont=sf,textfont=sf]{subfig}
\usepackage{fixltx2e}
\usepackage{txfonts}
\usepackage{stfloats}
\newtheorem{theorem}{\textbf{Theorem}} 
\newtheorem{conjecture}{\textbf{Conjecture}} 
 
\newtheorem{lemma}{\textbf{Lemma}}
\newtheorem{corollary}{\textbf{Corollary}}

\newtheorem{definition}{\textbf{Definition}}

\usepackage{lipsum}  
\usepackage{color}

\begin{document}
\title{Minimum Length Scheduling for Full Duplex Time-Critical Wireless Powered Communication Networks}
\author{Muhammad Shahid~Iqbal, Yalcin Sadi,~\IEEEmembership{Member,~IEEE}
        and~Sinem Coleri,~\IEEEmembership{Senior Member,~IEEE}
\thanks{Muhammad Shahid~Iqbal and Sinem Coleri are with the department of Electrical and Electronics Engineering, Koc University, Istanbul, Turkey, email: $\lbrace$miqbal16, scoleri$\rbrace$@ku.edu.tr. Yalcin Sadi is with the department of Electrical and Electronics Engineering, Kadir Has University, Istanbul, Turkey, email: yalcin.sadi@khas.edu.tr. This work is supported by Scientific and Technological Research Council of Turkey Grant $\#$117E241.}}
\maketitle
\begin{abstract}
Radio frequency (RF) energy harvesting is key in attaining perpetual lifetime for time-critical wireless powered communication networks (WPCNs) due to full control on energy transfer, far field region, small and low-cost circuitry. In this paper, we propose a novel minimum length scheduling problem to determine the optimal power control, time allocation and schedule subject to data, energy causality and maximum transmit power constraints in a full-duplex WPCN. We first formulate the problem as a mixed integer non-linear programming problem and conjecture that the problem is NP-hard. As a solution strategy, we demonstrate that the power control and time allocation, and the scheduling problems can be solved separately in the optimal solution. For the power control and time allocation problem, we derive the optimal solution by evaluating Karush-Kuhn-Tucker conditions. For the scheduling, we introduce a penalty function allowing reformulation of the problem as a sum penalty minimization problem. Upon derivation of the optimality conditions based on the characteristics of the penalty function, we propose two polynomial-time heuristic algorithms and a reduced-complexity exact algorithm employing smart pruning techniques. Via extensive simulations, we illustrate that the proposed heuristic schemes outperform the schemes for predetermined transmission order of users and achieve close-to-optimal solutions.
\end{abstract}
\IEEEpeerreviewmaketitle

\section{Introduction} \label{sec:intro}
Time critical wireless sensor networks have been widely used in emergency alert systems and cyber-physical systems due to many advantages, including easy installation and maintenance, low complexity and cost, and flexibility \cite{wncs_sinem, intravehicle_sinem}. Several studies have been conducted on minimizing the schedule length given the traffic demand and limited battery lifetime of the users in these networks \cite{minlength_sinem, wncs_ref194}. However, recent developments in energy harvesting technologies have the potential to provide perpetual energy, eliminating the need to replace batteries.  
Considering the advantages of having full control on energy transfer, high range and small form factor, radio frequency (RF) energy harvesting is the most suitable technology \cite{harvest_10}. 
The recent advances in the design of highly efficient RF energy harvesting hardware is expected to even further extend its usage \cite{RFEH_efficientDesign1,RFEH_TFET_design,RFEH_efficientDesign2}. 

RF energy harvesting networks have been previously studied in the context of simultaneous wireless information and power transfer (SWIPT) and wireless powered communication networks (WPCNs). In SWIPT, the access point transmits energy and data simultaneously to multiple receivers in the downlink. The trade-off between wireless information transmission capacity and wireless energy transmission efficiency of a single user has been analyzed for point-to-point transmissions considering additive white gaussian noise (AWGN) channels \cite{harvest_06}, flat-fading channels \cite{harvest_61}, co-located or separated energy harvester and information decoder setup \cite{harvest_08}, and a non-linear energy harvesting model \cite{harvest_new63}. 

SWIPT based multi-user systems mostly optimize the performance 
by incorporating maximization of the weighted energy transfer 
\cite{harvest_19}, throughput maximization \cite{harvest_new62} and 
total transmit power minimization \cite{harvest_59} as objective 
function constrained by minimum signal to noise ratio, data buffer limit and harvested power, respectively. The energy efficiency of 
these systems is studied by considering the power budget \cite{SWIPT_EE},
 artificial noise \cite{PCSI1}, relay based setup \cite{SWIPT_EE1} 
and co-variance channel state information feedback \cite{harvest_21}.
These studies assume the simultaneous transmission of energy and information without considering any scheduling. The scheduling of SWIPT based networks has been considered in a limited context in \cite{harvest_01_ref146} and \cite{harvest_01_ref169}. The time is divided into multiple slots. In each time slot, a single user is selected for information reception while energy is transferred to the remaining users. The scheduling algorithms are proposed for the selection of this single user in each time slot. 

In WPCN, the wireless users harvest energy from the access point in the downlink and then transmit data to the access point in the uplink. The first protocol proposed for WPCN, called \emph{harvest-then-transmit}, is based on dynamic time-division multiple access (TDMA) in a half-duplex framework. Each TDMA frame is divided into two non-overlapping variable length intervals used for the wireless energy transmission in downlink and information transmission of the users in uplink \cite{harvest_07}. The objective is to maximize the total throughput by optimally allocating the uplink and downlink transmission times. 
Since the objective of throughput maximization results in unfair achievable rates among different users, with the corresponding allocation substantially favoring near users with mostly better channel conditions, some of the later works have focused on alternative objective functions, such as maximization of minimum throughput \cite{harvest_04}, maximization of weighted sum rate of uplink information transmission \cite{harvest_50, harvest_44, harvest_55}, maximization of energy efficiency \cite{harvest_51} and minimization of schedule length \cite{harvest_sinem, harvest_elif}. Other studies, on the other hand, have included the usage of near users as relays by using some of their energy and time to relay information of the farther users \cite{harvest_39_ref17, harvest_39, harvest_56}. The order of information transmission in the uplink does not matter due to the non-overlapping characteristic of the wireless energy and information transmission, thus, no scheduling algorithm is required in these half-duplex systems. Although these studies impose a transmit power constraint on the access point for wireless energy transmission in the form of either the assignment of a constant value \cite{harvest_07, harvest_19, harvest_04, harvest_39_ref17, harvest_39, harvest_51, harvest_56} or constraint on its average and maximum value \cite{harvest_50, harvest_44, harvest_55}, no upper bound has been imposed on the transmit power of the users in their information transmission. Also, the initial battery level of the users are not considered in these works, except \cite{harvest_51}.

WPCNs recently started to incorporate full duplex technology with the goal of improving the amount of transferred energy by allowing the access point to simultaneously transfer wireless energy and receive information, and in some cases also enabling the concurrent reception of wireless energy and transmission of information at the users. The main challenge in full duplex systems is to mitigate the self-interference, where part of the transmitted signal is received by itself, thus interfering with the desired received signal. The recent advances in self-interference cancellation (SIC) techniques \cite{harvest_41_ref26, harvest_41_ref27} and their practical implementations \cite{harvest_30_ref19, harvest_30_ref26} placed full-duplex as one of the key transceiving techniques for 5G networks \cite{harvest_new66}. Full-duplex WPCN systems have been mostly formulated with the goal of maximizing the sum throughput by assuming either only access point operating in full-duplex mode \cite{harvest_30, harvest_40, harvest_50} or both access point and users operating in full-duplex mode \cite{harvest_50, harvest_41, harvest_new68}. \cite{harvest_30} assumes perfect self-interference cancellation, whereas \cite{harvest_40, harvest_41, harvest_50} include residual self-interference, proportional to the transmit power of the access point. 
Only, \cite{harvest_30} additionally considers the minimization of schedule length given the traffic demand of the links. In full-duplex systems, since the users can harvest energy during the transmission of other users, the order of transmission so scheduling of user transmissions is important. However, previous studies assume predetermined transmission order without considering any scheduling algorithm. Moreover, none of these studies consider any limitation on the transmit power of the users, while assuming either constant transmission power \cite{harvest_30, harvest_40} or a maximum power constraint \cite{harvest_50, harvest_41} for the access point. Furthermore, these studies assume that the energy required for the data transmission needs to be supplied by the wireless transfer, without considering the initial battery level of the users. 

The goal of this paper is to determine the optimal time allocation, power control, and scheduling with the objective of minimizing the schedule length subject to the traffic requirement, the maximum transmit power constraint, and the energy causality constraint of the users, for a time-critical WPCN. The original contributions of this paper are listed as follows:

\begin{itemize}
\item We propose a new optimization framework for a full-duplex WPCN, employing the maximum transmit power and energy causality constraints and considering the initial battery levels and full duplex energy harvesting capability of the users.
\item We characterize Minimum Length Scheduling Problem ($\cal{MLSP}$) aiming at determining the optimal power control, time allocation and scheduling with the objective of minimizing the completion time of the schedule subject to data, energy causality and maximum transmit power constraints of the users. We formulate the problem mathematically as a mixed integer nonlinear programming (MINLP) problem, which is non-convex and thus generally difficult to solve for a global optimum. We further conjecture that $\cal{MLSP}$ is NP-hard based on the reduction of Single Machine Scheduling Problem ($\cal{SMSP}$) which is proven to be NP-hard \cite{jiang2013single} to $\cal{MLSP}$. Then, we propose a solution framework based on the decomposition of $\cal{MLSP}$ to optimal power control and time allocation, and optimal scheduling problems.
\item We formulate the power control and time allocation problem as a convex optimization problem and derive the optimal solution in closed-form by evaluating the Karush-Kuhn-Tucker (KKT) conditions.
\item For the scheduling problem, we introduce a penalty function, defined as the difference between the actual and minimum possible transmission time of a user. This allows consideration of the schedule length minimization objective as the minimization of the sum of the penalties of the users. By exploiting the characteristics of the penalty function, we analyse the optimality conditions of the scheduling. Based on the derived optimality conditions, we propose two polynomial-time heuristic algorithms and one exact exponential-time algorithm with significantly reduced complexity based on smart enumeration techniques.
\item We evaluate the performance of the proposed scheduling algorithms for various parameters, including transmit power of the access point, maximum transmit power of the users, and network size, in comparison to the optimal solution and conventional schemes proposed for the minimum length scheduling of users with a predetermined transmission order. We illustrate that the proposed polynomial time heuristic algorithms perform very close to optimal while outperforming previously proposed algorithms significantly.
\end{itemize}

The rest of the paper is organized as follows. Section \ref{sec:system} describes the WPCN model and assumptions used in the paper. Section \ref{sec:mlsp} presents the mathematical formulation of the minimum length scheduling problem, investigates its complexity and introduces our solution strategy based on the decomposition of the problem. Section \ref{sec:power} presents the optimal power control and time allocation problem and derives its optimal solution. Section \ref{sec:scheduling} presents the optimal scheduling problem, analyzes its optimality conditions, proposes one exact reduced complexity exponential-time algorithm and two polynomial-time heuristic algorithms. Section \ref{sec:simulation} evaluates the performance of the proposed scheduling schemes. Conclusions are presented in Section \ref{sec:conclusion}.

\section{System Model and Assumptions} \label{sec:system}

The system model and assumptions are described as follows:

\begin{enumerate}
\item The WPCN architecture consists of a hybrid access point (HAP) and N users; i.e., sensors and machine type communication (MTC) devices. Both the HAP and the users are equipped with one full-duplex antenna. Full duplex antennas are used for simultaneous wireless energy transfer on the downlink from the HAP to the users and data transmission on the uplink from the users to the HAP. We have considered single antenna HAP to decrease the complexity of the HAP and algorithm design at the first step of the study. 

\item We consider Time Division Multiple Access (TDMA) as medium access control protocol for the uplink data transmission from the users to the HAP. The time is partitioned into scheduling frames, which are further divided into variable-length slots each allocated to a particular user.
The energy transfer from the HAP to the users continues throughout the frame. Each user can use the energy it harvests from the beginning of the frame till the end of its transmission, including both its own dedicated time slot and the time slots allocated for the previously scheduled users. The energy harvested by a user after its dedicated slot can be stored in the battery for possible usage in the subsequent scheduling frames.

\item The HAP is equipped with a stable energy supply and continuously transfers wireless energy with a constant power $P_h$. The users operate in harvest-use-store (HUS) mode in which they prioritize the harvested energy for data transmissions over the energy stored in the battery \cite{HUS}. HUS mode is more energy efficient compared to harvest-store-use (HSU) mode, where the harvested energy is first stored in a battery before its subsequent use \cite{HSU} due to severe storage loss for batteries with low storage efficiency. Besides, HUS mode introduces negligible processing delay due to the direct use of the harvested energy compared to HSU mode. In HUS mode, each user $i$ harvests energy from the HAP during the entire scheduling frame. If the harvested energy is larger than or equal to the energy required for data transmission, user $i$ uses it directly for data transmission and stores the excess energy in a rechargeable battery with an initial energy $B_i$ at the beginning of the scheduling frame. Otherwise, it uses all harvested energy and a partial energy from the battery.

\item The downlink and uplink channel gains, denoted by $h_i$ and $g_i$ for user $i$, respectively, are assumed to be different. 
Both channels are assumed to be block-fading, i.e., the channel gains remain constant over the scheduling frame, as commonly used in previous WPCN studies  \cite{harvest_39, harvest_51, harvest_56}. 
We assume the availability of the channel information to focus on the joint optimization of power, time allocation and scheduling in the first step of the study. The perfect channel state information is a very common assumption in recent literature on WPCN \cite{PCSI2,PCSI1}. The time and energy cost of the collection of the channel state information can be considered negligible for a low mobility network \cite{PCSI2}.
\item We assume a realistic non-linear energy harvesting model based on logistic function \cite{NLEH_01}, which performs close to the experimental results proposed in \cite{NLEH_Prac_01}. The energy harvesting rate for user $i$ is given by 
\begin{equation} \small
C_i=\dfrac{P_s[\Psi_i-\Omega_i]}{1-\Omega_i},
\end{equation}
where $\Omega_i=\dfrac{1}{1+e^{A_iB_i}}$ is a constant to ensure zero-input zero-output response, $P_s$ is the maximum harvested power during saturation and $\Psi_i$ is the logistic function related to user $i$ given by:
\begin{equation} \small
\Psi_i=\dfrac{1}{1+e^{-A_i(h_iP_h-B_i)}},
\end{equation}
where $A_i$ and $B_i$ are the positive constants related to the non-linear charging rate with respect to the input power and turn-on threshold, respectively. For a given energy harvesting circuit, the parameters $P_s$, $A_i$ and $B_i$ can be determined by curve fitting.
\item We assume that user $i$ has to transmit $D_i$ bits over the scheduling frame. 

\item We use continuous rate transmission model, in which Shannon's channel capacity formulation for an AWGN wireless channel is used in the calculation of the maximum achievable rate as a function of Signal-to-Interference-plus-Noise Ratio (SINR) as $ x_i = W \log_{2}(1+k_iP_i)$, where $x_i$ is the transmission rate of user $i$, $P_i$ is the transmission power of user $i$, $W$ is the channel bandwidth, and $k_i$ is defined as $g_i/(N_0 W+\beta P_h)$, in which the term $\beta P_h$ is the power of self interference at the HAP and $N_0$ is the noise power density. Although the networks are generally restricted to support discrete rates, the continuous rate assumption is conventionally used in most of the studies in the literature \cite{harvest_new62,harvest_07,harvest_30,harvest_50}.
 \item We use continuous power model in which the transmission power of a user can take any value below a maximum level $P_{max}$ imposed to avoid the interference to nearby systems.
\end{enumerate}

\section{Minimum Length Scheduling Problem} \label{sec:mlsp}

In this section, we introduce the minimum length scheduling problem referred as $\cal{MLSP}$. We first present the mathematical formulation of $\cal{MLSP}$ as an optimization problem and investigate its complexity. Then, we provide the solution strategy followed in the subsequent sections.
\vspace{-5mm}
\subsection{Mathematical Formulation}
The joint optimization of the time allocation, power control and scheduling with the objective of minimizing the schedule length given the traffic demands of the users while considering realistic transmission model for a full-duplex system is formulated as follows:

$\cal{MLSP}$:
\begin{subequations} \label{opt_problem} \small
\begin{align}
& \textit{minimize}
& & \sum_{i=0}^{N}\tau_i \label{obj}\\
& \textit{subject to}
& & W\tau_i\log_2\Big(1+k_iP_i\Big)\geq D_i, \hspace*{0.1cm} i \in \{1,...,N\} \label{datarate}\\
&&& B_i+C_i  \sum_{j=0}^{N}a_{ji}\tau_j +(C_i - P_i)\tau_i\geq 0, \hspace*{0.1cm} i \in \{1,...,N\} \label{energyharvesting}\\
&&& a_{ij}+a_{ji}=1, \hspace*{0.1cm} i<j, i,j \in \{1,...,N\} \label{ordering}\\
&&& P_i\leq P_{max}, \hspace*{0.1cm} i \in \{1,...,N\}  \label{pmax} \\
& \textit{variables}
& & P_i \geq 0, \hspace*{0.1cm} \tau_i\geq 0, \hspace*{0.1cm} a_{ij} \in \{0,1\}, \hspace*{0.1cm} i,j \in \{1,...,N\} .\label{pcp1_vars}
\end{align}
\end{subequations}

The variables of the problem are $P_i$, the transmit power of user $i$; $\tau_i$, the transmission time of user $i$, and $a_{ij}$, binary variable that takes value $1$ if user $i$ is scheduled before user $j$ and $0$ otherwise. In addition, $\tau_0$ denotes an initial waiting time duration during which all the users only harvest energy without transmitting any information \cite{harvest_50}  .

The objective of the optimization problem is to minimize the schedule length as given by Eq. (\ref{obj}). Eq. (\ref{datarate}) represents the constraint on satisfying the traffic demand of the users. Eq. (\ref{energyharvesting}) gives the energy causality constraint: The total amount of available energy, including both the initial energy and the energy harvested until and during the transmission of a user, should be greater than or equal to the energy consumed during its transmission. Eq. (\ref{ordering}) represents the scheduling constraint, stating that if user $i$ transmits before user $j$, user $j$ cannot transmit before user $i$. Eq. (\ref{pmax}) represents the maximum transmit power constraint. 

This optimization problem is a MINLP thus difficult to solve for a global optimum \cite{opt_book}. 
\vspace{-5mm}

\subsection{Complexity Analysis} \label{sec:complexity}

In this section, we provide the conjecture on the NP hardness of $\cal{MLSP}$, as strongly supported by the evidence provided thereafter, since a formal proof demonstrating the reduction of an NP complete problem to the decision version of $\cal{MLSP}$ in polynomial time cannot be provided.
\begin{conjecture}
$\cal{MLSP}$ is NP-hard.
\end{conjecture}

First strong evidence supporting the conjecture arises from the analogy between $\cal{MLSP}$ and the single-machine scheduling problem studied in \cite{jiang2013single}, denoted by $\cal{SMSP}$.

$\cal{SMSP}$ is characterized as follows:\\
$\cal{SMSP}$: Given a set ${\cal{J}} = \lbrace J_1, ..., J_n\rbrace$ of jobs, a normal processing time $p_i \geq 0$, a learning coefficient $\alpha_i\leq 0$, an actual processing time function $p_{ir}$ for each job $J_i$, and a positive integer $y$, is there a schedule with makespan $C \leq y$, where makespan is defined as the completion time of the last scheduled job? The actual processing time $p_{ir}$ of a job $J_i$ scheduled in the $r^{th}$ position is defined as
\begin{equation} \label{eq:learning1}
p_{ir} = \left( 1 + p_{[1]}^A + p_{[2]}^A + ... + p_{[r-1]}^A \right)^{\alpha_i} p_i,
\end{equation}
where $p_{[k]}^A$ is the actual processing time of the job scheduled in the $k^{th}$ position with $p_{[1]}^A = p_{[1]}$, where $p_{[k]}$ is the normal processing time of the job scheduled in the $k^{th}$ position. Note that normal processing time is the processing time of a job if it is scheduled first; i.e., $r=1$.

$\cal{SMSP}$ is proven to be NP-complete by reducing the well-known NP-complete Partition problem to $\cal{SMSP}$ in polynomial-time in \cite{jiang2013single}. Consequently, minimizing the makespan for a set of jobs with the characteristics detailed in the characterization of $\cal{SMSP}$ is NP-hard. Note that $\cal{SMSP}$ is the decision version of the makespan minimization problem.

The processing time given in Eq. (\ref{eq:learning1}) is based on a \textit{time and job-dependent learning} model such that as the sum of the actual processing time of the previously completed jobs increases, a particular job is processed in shorter time with respect to its normal processing time. We can assume that the schedule starts at time $t=0$ without loss of generality. Then, based on Eq. (\ref{eq:learning1}), the actual processing time of a job $J_i$ scheduled at time $t$ can be given as
\begin{equation} \label{eq:learning2}
p_{i}(t) = \left( 1 + t \right)^{\alpha_i} p_i.
\end{equation}
Now consider $\cal{MLSP}$. Jobs, processing times and makespan in $\cal{SMSP}$ correspond to users, transmission times and schedule length in $\cal{MLSP}$, respectively. Moreover, the effect of learning on the processing time of a job in $\cal{SMSP}$ corresponds to the effect of energy harvesting on the transmission time of a user in $\cal{MLSP}$. Let $\tau_{i}(t)$ be the transmission time of user $i$ scheduled at time $t$. Due to energy harvesting, user $i$ can complete its transmission in less time as the scheduling time increases; i.e., $\tau_{i}(t)$ is a monotonically decreasing function of $t$, since it will be able to transmit with larger transmit power unless constrained by maximum transmit power level $P_{max}$. Therefore, for a sufficiently large $P_{max}$ value, such that no user can afford to transmit with a transmit power greater than or equal to $P_{max}$ in the optimal solution, $\tau_{i}(t)$ can be formulated as 
\begin{equation} \label{eq:learning3}
\tau_{i}(t) = \left( 1 + t \right)^{f_i(t)} \tau_i,
\end{equation}
where $f_i(t)$ is a time dependent function representing the dependency of the scheduling time on the energy harvesting rate of user $i$ such that $f_i(t)\leq 0$ for $t\geq 0$ and $\tau_i$ is the transmission time of user $i$ if it is scheduled first, i.e., $\tau_i=\tau_{i}(0)$. Then, the learning effect based processing time model given by Eq. (\ref{eq:learning2}) is an instance of the transmission time model in Eq. (\ref{eq:learning3}) for $f_i(t) = \alpha_i$; i.e., $f_i(t)$ is a constant. Hence, considering the additional complexity introduced by the time-dependence in the problem, the decision version of $\cal{MLSP}$ can be considered at least as \textit{hard} as $\cal{SMSP}$.

Second evidence supporting the Conjecture exists in our study on the discrete-rate transmission model based version of $\cal{MLSP}$ \cite{iqbal2020minimum} for which we formally prove the NP-hardness.

\subsection{Solution Strategy}

As the mathematical formulation and the complexity analysis presented in previous sections suggest, it is \textit{difficult} to solve $\cal{MLSP}$ for a global optimum, i.e., finding a global optimum requires algorithms with exponential complexity. Such optimal algorithms are intractable even for moderate problem sizes. In this paper, we present a solution framework to overcome this intractability based on the decomposition of the optimal power and time allocation and the optimal scheduling problems as described below:

\begin{itemize}
\item For a given scheduling order of the users, $\cal{MLSP}$ requires determining the optimal power and time allocation of the users with minimum schedule length while considering their data, maximum transmit power and energy causality constraints. We first show that this problem is a convex optimization problem suggesting that it is polynomial-time solvable. Then, we provide the optimal solution based on the analysis of the KKT conditions.
\item Determining the optimal power and time allocation for a given scheduling order reduces $\cal{MLSP}$ to the optimization of the scheduling order. We first introduce penalty function defined as the difference between the actual and minimum possible transmission time of a user and demonstrate the equivalence between schedule length minimization objective and the minimization of the sum of the penalties of the users. Then, based on the optimality conditions derived using the penalty function, we propose two polynomial-time heuristic algorithms that perform very close to optimal. Furthermore, we propose an exact exponential-time algorithm with reduced complexity based on smart pruning techniques.
\end{itemize}

\section{Optimal Power Control}  \label{sec:power}

In this section, we are interested in determining the optimal power control and time allocation to minimize the schedule length for a given transmission order of a set of users; i.e. $a_{ij}$'s are given in $\cal{MLSP}$. 

We first illustrate that inclusion of $\tau_0$, the initial waiting time, is not actually needed.

\begin{lemma} \label{lemma_tau}
In the optimal solution of $\cal{MLSP}$, $\tau_0=0$.
\end{lemma}
\vspace{-5mm}
\begin{proof}
Suppose that $\tau_0^*>0$ is the optimal energy harvesting time, $\lbrace\tau_1^*, \tau_2^*, ..., \tau_N^* \rbrace$ and $\lbrace P_1^*, P_2^*, ..., P_N^* \rbrace$ are the sets of optimal transmission times and transmit powers, respectively. Then, the energy consumed by user $1$ is $E_1^*= \tau_1^* P_1^*$ and $E_1^* \leq B_1 + C_1 (\tau_1^* + \tau_0^*)$ due to the energy causality constraint. Now consider that, instead of waiting for a duration $\tau_0^*$, user $1$ transmits the same amount of data in a time slot with length $\tau_1^{'}=\tau_1^* + \tau_0^*$ with transmit power $P_1^{'}< P_1^*$. Since the energy required for the transmission of a fixed amount of data $D_i$ is a monotonically increasing function of transmit power $P_i$; i.e., $E_i=\frac{D_i}{W \log_{2}(1+k_iP_i)} P_i$,  $E_1^{'}< E_1^*$. Then, the energy causality constraint is not violated since $E_1^{'}< E_1^* \leq B_1 + C_1 (\tau_1^* + \tau_0^*) = B_1 + C_1 \tau_1^{'}$. 
This is a contradiction.
\end{proof}

Note that $\tau_0$ can be interpreted as a delay in the transmission of a user which can complete its transmission without this delay in the same amount of time using less energy. Then, considering the transmission of an arbitrary user in a schedule, we have the following corollary.

\begin{corollary} \label{corollary_tau}
Delaying the transmission of a user to harvest more energy also delays the completion time of the transmission of the user.
\end{corollary}

Next, we mathematically formulate the power control and time allocation problem, denoted by $\cal{PCP}$. Without loss of generality, we assume that user $i$ transmits in time slot $i$. For brevity, we present the formulation with a variable transformation to illustrate its convexity as follows:

$\cal{PCP}$:
\begin{subequations} \label{opt_problem2} \small
\begin{align}
& \textit{minimize}
& & \sum_{i=1}^{N}\tau_i \label{obj2}\\
& \textit{subject to}
& & D_i-W\tau_i\log_2\Big(1+k_i\alpha_i/\tau_i\Big)\leq 0, \hspace*{0.1cm} i \in \{1,...,N\} \label{datarate2}\\
&&& \alpha_i\leq \tau_i P_{max}, \hspace*{0.1cm} i \in \{1,...,N\} \label{pmax2}\\
&&& \alpha_i-B_i-C_i\sum_{j=1}^{i}\tau_j \leq 0, \hspace*{0.1cm} i \in \{1,...,N\} \label{energyharvesting2}\\
& \textit{variables}
& & \alpha_i \geq 0, \hspace*{0.1cm} \tau_i\geq 0, \hspace*{0.1cm} i \in \{1,...,N\}\label{pcp2_vars}
\end{align}
\end{subequations}
where $\alpha_i\triangleq P_i \tau_i$ denotes the energy consumption of user $i$ during its transmission with transmit power $P_i$ in the time slot $i$ with length $\tau_i$. The variables of the problem are $\tau_i$, the transmission time of user $i$, and $\alpha_i$, the energy consumption of user $i$ during its transmission. It can be observed that $\cal{PCP}$ is a convex optimization problem considering the linearity of the objective, the convexity of the constraint (\ref{datarate2}) and the affineness of the constraints (\ref{pmax2}) and (\ref{energyharvesting2}).

Since $\cal{PCP}$ is a convex optimization problem, it can be solved by evaluating its KKT conditions, which specify necessary and sufficient conditions for an optimal solution of a convex optimization problem; i.e., any feasible point satisfying the KKT conditions is a global optimum point \cite{ex24}.

KKT conditions of $\cal{PCP}$ are as follows:
 \begin{subequations} \label{eqs:gra} \small
\begin{align}
&1+W\lambda_{1i}\Bigg(\frac{k_i\alpha_i}{\tau_i+k_i\alpha_i}-\log_2\Big(1+k_i\alpha_i/\tau_i\Big)\Bigg)-\lambda_{2i}P_{max}-\lambda_{3i}C_i=0, \nonumber \\
& i \in \{1,...,N\} \label{gra1}\\
& \lambda_{2i}-\frac{\lambda_{1i}Wk_i\tau_i}{\tau_i+k_i\alpha_i}+\lambda_{3i}=0, \hspace*{0.1cm} i \in \{1,...,N\} \label{gra2}
\end{align}
\end{subequations} 

\begin{subequations} \label{eqs:comp} \small
\begin{align}
&\lambda_{1i}\bigg(D_i-W\tau_i\log_2\Big(1+k_i\alpha_i/\tau_i\Big)\bigg)=0, \hspace*{0.1cm} i \in \{1,...,N\} \label{comp1}\\
&\lambda_{2i}\Big(\alpha_i-\tau_iP_{max}\Big)=0, \hspace*{0.1cm} i \in \{1,...,N\} \label{comp2}\\
&\lambda_{3i}\Big(\alpha_i-B_i-C_i\sum_{j=1}^{i}\tau_j\Big)=0, \hspace*{0.1cm} i \in \{1,...,N\}\label{comp3}
\end{align}
\end{subequations}
where Eqs. (\ref{gra1})-(\ref{gra2}) and Eqs. (\ref{comp1})-(\ref{comp3}) represent the gradient and complementary slackness conditions, respectively, and $\lambda_{ji}\geq 0$ is the KKT multiplier associated with the $j^{th}$ constraint of the $i^{th}$ user, $j \in \{1,2,3\}$ corresponding to constraints (\ref{datarate2})-(\ref{energyharvesting2}), respectively. 
\begin{lemma} \label{lemma:data_cons}
In an optimal solution of $\cal{PCP}$, the constraint $(\ref{datarate2})$ must be satisfied with equality.
\end{lemma}

\begin{proof}
Suppose that in an optimal solution of $\cal{PCP}$, the constraint $(\ref{datarate2})$ is not satisfied with equality. For an optimal solution, KKT conditions given in Eqs. (\ref{eqs:gra})-(\ref{eqs:comp}) should be satisfied. Since constraint $(\ref{datarate2})$ is not binding, $\lambda_{1i}=0$ by Eq. (\ref{comp1}). Then, by Eq. (\ref{gra2}), $\lambda_{2i}=0$ and $\lambda_{3i}=0$ since $\lambda_{ji}\geq 0$. However, for $\lambda_{1i}=\lambda_{2i}=\lambda_{3i}=0$, Eq. (\ref{gra1}) is violated. This is a contradiction.
\end{proof}

\begin{lemma} \label{lemma:power_cons}
In an optimal solution of $\cal{PCP}$, either constraint $(\ref{pmax2})$ or $(\ref{energyharvesting2})$ must be satisfied with equality.
\end{lemma}

\begin{proof}
Suppose that none of the constraints $(\ref{pmax2})$ and $(\ref{energyharvesting2})$ are binding. Then, $\lambda_{2i}=0$ and $\lambda_{3i}=0$ by Eqs. (\ref{comp2}) and (\ref{comp3}). Then, either Eq. (\ref{gra1}) or Eq. (\ref{gra2}) is violated for any $\lambda_{1i}\geq 0$, which violates KKT conditions. On the other hand, suppose that both constraints $(\ref{pmax2})$ and $(\ref{energyharvesting2})$ are binding. Then, due to Lemma \ref{lemma:data_cons}, these two constraints together with constraint $(\ref{datarate2})$ specify an overdetermined system of equations for which a solution does not exist. This is a contradiction.
\end{proof}

Now, we provide the optimal power control and time allocation for $\cal{PCP}$.

\begin{theorem} \label{thm:opt_power}
In the optimal solution of $\cal{PCP}$, transmit power $P_i$ of a user $i$ is given by

\begin{equation} \label{eq:opt_pow_1} \small
P_i = \min \left\lbrace P_{max} ,  \left(	\frac{-1}{k_i} - \frac{W \varepsilon_i}{D_i \ln 2}	{\cal{W}} ( \frac{-D_i \ln 2 }{W k_i \varepsilon_i} e^{\upsilon_i} ) \right) \right\rbrace,
\end{equation}
where ${\cal{W}}(\cdot)$ is the Lambert function \cite{ex25},
\begin{equation} \label{eq:opt_pow_2} \small
\varepsilon_i = B_i + C_i \sum_{j=1}^{i-1} \tau_j,
\end{equation}
\begin{equation} \label{eq:opt_pow_3} \small
\upsilon_i = \frac{-C_i D_i \ln 2}{W \varepsilon_i} - \frac{D_i \ln 2}{W k_i \varepsilon_i}.
\end{equation}
Then, the optimal time allocation $\tau_i$ of user $i$ is given by

\begin{equation} \label{eq:opt_pow_4} \small
\tau_i = D_i / \left( W \log_2 (1 + k_i P_i)\right).
\end{equation}

\end{theorem}

\begin{proof}
By Lemmas \ref{lemma:data_cons} and \ref{lemma:power_cons}, either constraints $(\ref{datarate2})$ and $(\ref{pmax2})$ or $(\ref{datarate2})$ and $(\ref{energyharvesting2})$ are satisfied with equality in an optimal solution. Let $P_i^*$ be the transmit power satisfying constraints $(\ref{datarate2})$ and $(\ref{energyharvesting2})$ with equality. Note that $P_i^*$ may satisfy or violate constraint $(\ref{pmax2})$. Then,
\begin{equation} \label{1}\small
\displaystyle D_i-W\tau_i\log_2\Big(1+k_i\alpha_i/\tau_i\Big) = 0 
\end{equation}
\begin{equation} \label{2}\small
\displaystyle \alpha_i-B_i-C_i\sum_{j=1}^{i-1}\tau_j-C_i \tau_i = 0
\end{equation}
where $\alpha_i = P_i^* \tau_i$. Define $\varepsilon_i = B_i + C_i \sum_{j=1}^{i-1} \tau_j$. Then, Eq. (\ref{2})  can be rearranged as
\begin{equation} \label{3} \small
\displaystyle \alpha_i=\varepsilon_i+C_i\tau_i
\end{equation}
Inserting Eq. (\ref{3}) into Eq. (\ref{1}), Eq. (\ref{1}) can be rearranged as
\begin{equation} \label{4} \small
\displaystyle e^{\bigg(\dfrac{D_i \ln 2}{W\tau_i}\bigg)}=1+k_i(\varepsilon_i+C_i\tau_i)/ \tau_i
\end{equation}
We can represent Eq. (\ref{4}) in the form $Y=X e^X$ as
\begin{equation} \label{5} \small
\displaystyle \dfrac{-D_iln2}{Wk_i \varepsilon_i}e^{\upsilon_i}=\bigg(\dfrac{-D_i ln2}{W\tau_i}+\upsilon_i\bigg)e^{\bigg(\dfrac{-D_i ln2}{W\tau_i}+\upsilon_i\bigg)}
\end{equation}
where $\upsilon_i$ is given by Eq. (\ref{eq:opt_pow_3}). Solution to $Y=X e^X$ is $X={\cal{W}}(Y)$. Then, omitting some steps for brevity, $\tau_i$ is obtained as
\begin{equation} \label{7} \small
\displaystyle \tau_i=-1/ \left( 
\frac{1}{k_i\varepsilon_i}+\frac{C_i}{\varepsilon_i}-\frac{W}{D_i \ln 2}{\cal{W}}(\frac{-D_i \ln 2}{Wk_i\varepsilon_i}e^{\upsilon_i})\right) 
\end{equation}
and since $P_i^*=\dfrac{\varepsilon_i}{\tau_i}+C_i$ from Eq. (\ref{3}), $P_i^*$ is 
\begin{equation} \small
\small P_i^* = 	\frac{-1}{k_i} - \frac{W \varepsilon_i}{D_i \ln 2}	{\cal{W}} ( \frac{-D_i \ln 2 }{W k_i \varepsilon_i} e^{\upsilon_i} ) 
\end{equation}
Since the energy consumed during transmission of a fixed amount of data is a monotonically increasing function of the transmit power, as given in the proof of Lemma $\ref{lemma_tau}$,  any power allocation $P_i > P_i^*$ violates the energy causality constraint $(\ref{energyharvesting2})$; whereas, any power allocation $P_i \leq P_i^*$ satisfies it; i.e., $P_i^*$  is the maximum transmit power satisfying the energy causality constraint. Then, if $P_i^*$ is feasible considering constraint $(\ref{pmax2})$; i.e., $P_i^* \leq P_{max}$, $P_i=P_i^*$ is optimal. Otherwise, if $P_i^*$ violates constraint $(\ref{pmax2})$; i.e., $P_i^* > P_{max}$, since an optimal solution should satisfy either constraints $(\ref{datarate2})$ and $(\ref{pmax2})$ or $(\ref{datarate2})$ and $(\ref{energyharvesting2})$ with equality and the latter case does not yield a feasible power allocation, constraint $(\ref{pmax2})$ is satisfied with equality in the optimal solution. Hence, $P_i=P_{max}$ is optimal. Therefore, optimal power allocation $P_i$ is given by the minimum of $P_i^*$ and $P_{max}$; i.e., $P_i=\min\lbrace P_i^*, P_{max}\rbrace$. Besides, by Lemma \ref{lemma:data_cons}, the optimal time allocation $\tau_i$ is given by $\tau_i = D_i / \left( W \log_2 (1 + k_i P_i)\right)$.
\end{proof}

\section{Scheduling}  \label{sec:scheduling}

The goal of this section is to determine the optimal schedule; i.e., the transmission order of the users, in order to minimize the length of the schedule. In Section \ref{sec:power}, the optimal time allocation and power control have been determined for a given schedule. However, optimizing the schedule can further decrease the schedule length.  
A straightforward solution to find the optimal schedule would be a brute-force search algorithm that enumerates all possible orderings of the users and then determines the one with the minimum length. However, such an algorithm has exponential complexity, which makes it computationally intractable for even a medium size network. Hence, fast and scalable solutions are required. 

In this section, we propose two polynomial-time heuristic algorithms and an exponential-time optimal algorithm with significantly reduced complexity. Next, we first investigate the optimality conditions for a minimum length schedule and then present the algorithms.

Let $s_i$ and $e_i(s_i)$ be the starting and ending time of the transmission of user $i$, respectively.

\begin{definition} \label{def_pen}
The penalty function of user $i$, $\rho_i(s_i): [0, \infty) \rightarrow [\rho_i^{max}, 0]$, is defined as the difference between the actual transmission time and the minimum possible transmission time, and formulated as $\rho_i(s_i)= e_i(s_i) - s_i - t_i^{min}$, where $t_i^{min}$ is the minimum possible transmission time of user $i$ corresponding to maximum transmit power $P_i= P_{max}$ and $\rho_i^{max} = e_i(0) - t_i^{min}$ is the maximum penalty for user $i$ when it is scheduled first.
\end{definition}

\begin{figure}[t] 
 \centering
\includegraphics[width= 0.65 \linewidth]{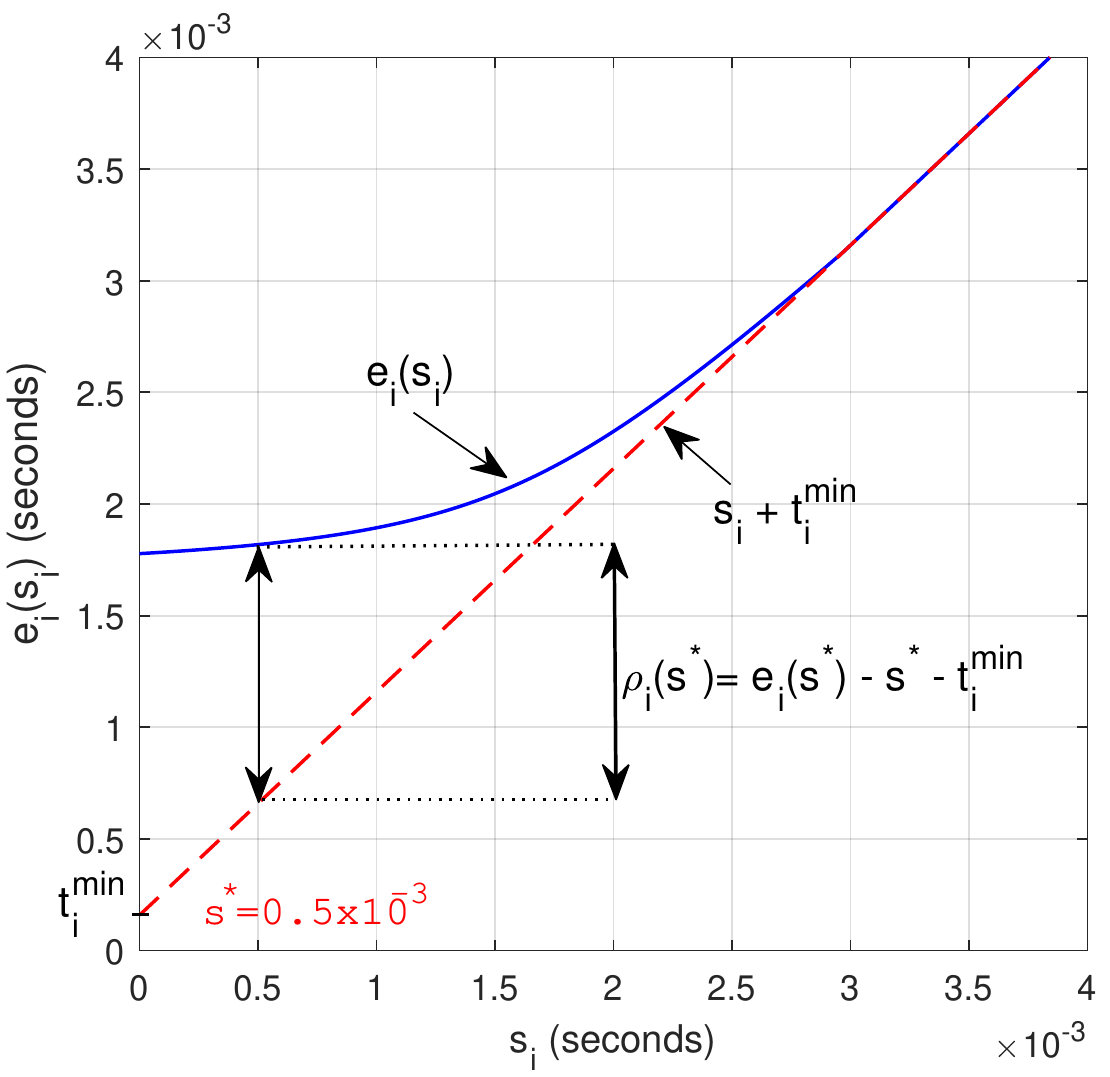}
\caption{Illustration of penalty function $\rho_i(s_i)$. $s_i+t_i^{min}$ is the lower bound for $e_i(s_i)$ while $e_i(s_i)$ becomes equal to $s_i+t_i^{min}$ for the first time when user $i$ can afford to transmit with $P_i=P_{max}$.} \label{fig:pen1}
\end{figure}

Penalty function $\rho_i(s_i)$ of user $i$ is illustrated in Fig. \ref{fig:pen1}. $\rho_i(s_i)$ decreases as a function of $s_i$ until it becomes $0$ when user $i$ can afford to transmit with $P_i=P_{max}$ for the first time. This monotonicity characteristic of the penalty function is stated in the following lemma.

\begin{lemma} \label{lemma_pen1}
The penalty function $\rho_i (s_i)$ is a non-increasing function of $s_i$.
\end{lemma}
\begin{proof}
By Lemma \ref{lemma:power_cons}, any user $i$ should either transmit with $P_i=P_{max}$ or consume all its energy during transmission. Let $s_i^{'}$ be the earliest time instant at which user $i$ transmits with $P_{max}$. For $s_i<s_i^{'}$, user $i$ consumes all its energy during transmission. As $s_i$ increases up to $s_i^{'}$, since the harvested energy increases, the energy consumed by the user increases so the transmit power increases up to $P_{max}$. Therefore, the transmission time decreases, resulting in a decrease in the penalty function. On the other hand, for any starting time $s_i\geq s_i^{'}$, the user transmits with $P_{max}$. Since the transmission time is equal to the minimum possible transmission time $t_i^{min}$ at transmit power $P_{max}$, the penalty is equal to $0$ for $s_i\geq s_i^{'}$.
\end{proof}

\begin{lemma} \label{lemma_equal}
For $\cal{MLSP}$, the objective of minimizing the schedule length $ \left( min\sum_{i=1}^{N}\tau_i\right) $ is equivalent to minimizing the sum of the penalties $ \left( min\sum_{i=1}^{N} \rho_i(s_i)\right) $.
\end{lemma}

\begin{proof}
By definition of the penalty function, 
\begin{equation} \label{penalty_}\small
min\sum_{i=1}^{N} \rho_i(s_i) = min\sum_{i=1}^{N} e_i(s_i) - s_i - t_i^{min} 
\end{equation}
Since $t_i^{min}$ is constant for any user $i$, $t_i^{min}$ can be removed from the objective function. Then, since $\tau_i = e_i(s_i) - s_i$, the objective function in Eq. (\ref{penalty_}) is reduced to $min\sum_{i=1}^{N} \tau_i$.
\end{proof}

\begin{theorem} \label{thm_pen}
If $\rho_i(0)=0$ for user $i$, then there exists an optimal solution to $\cal{MLSP}$ in which user $i$ is scheduled first with transmission time $\tau_i=t_i^{min}$.
\end{theorem}

\begin{proof}
Suppose that there exists exactly one optimal schedule $S^*$ with length $L^*$ in which user $i$ with $\rho_i(0)=0$ is not scheduled 
first. Denote the set of users scheduled before and after user $i$ by $U_b$ and $U_a$, respectively, and the penalty of each user $j$ by    
$\rho_j^*$. Note that $\rho_i^*=0$ based on Lemma \ref{lemma_pen1}. Now, consider that schedule $S^*$ is updated such that user $i$ is scheduled first and the scheduling order of the other users remain the same. Denote the resulting schedule by $S^{'}$ with length $L^{'}$ and the penalty of each user $j$ in the schedule by   
$\rho_j^{'}$. 
The reallocation of user $i$ will delay the starting time of first scheduled user in $U_b$ by $t_i^{min}$. Then, by Corollary \ref{corollary_tau}, the starting time and the corresponding ending time of each user $k\in U_b$ will be delayed consecutively. Therefore,  for any user $k\in U_b$, $\rho_k^{'}\leq \rho_k^*$ since the penalty is a nonincreasing function of starting time as given by Lemma $\ref{lemma_pen1}$. Hence, $\sum_{n\in \lbrace U_b \cup i\rbrace} \rho_n^{'} \leq \sum_{n\in \lbrace U_b \cup i\rbrace} \rho_n^* $. Then, the transmissions of user $i$ and users in $U_b$ in schedule $S^{'}$ will be completed either earlier or at the same time compared to schedule $S^*$. Consequently, the starting time and the corresponding ending time of the transmission of a user $l\in U_a$ will either decrease or remain the same, by Corollary \ref{corollary_tau}. Hence, the schedule length which is equal to the ending time of the last scheduled user in $U_a$ will not increase; i.e., $L^{'} \leq L^*$. This is a contradiction.
\end{proof}

Theorem \ref{thm_pen} can be interpreted that at time $t=0$, it is optimal to schedule a user $i$ with zero penalty if such a user exists initially. Consequently, at any time $t>0$, after the completion of the ongoing transmission, it is still optimal to schedule a user with zero penalty among the remaining unallocated users since making a scheduling decision on minimizing the schedule length at time $t>0$ requires minimizing the sum of the penalties of the remaining unallocated users. Then, we have the following corollary of Theorem $\ref{thm_pen}$ presenting an optimal online scheduling policy.

\begin{corollary} \label{corollary_opt1}
For any scheduling policy, at any time instant $t$,
\begin{enumerate}
\item It is optimal to schedule a user $i$ with $\rho_i(t)=0$ next after the currently scheduled user completes its transmission.
\item It is optimal to schedule a user $i$ that can afford to start its transmission at time $t$ and complete it using maximum transmit power $P_i=P_{max}$ without violating the energy causality constraint, next after the currently scheduled user completes its transmission.
\end{enumerate}
\end{corollary}
Next, we introduce the scheduling algorithms based on the foregoing optimality analysis.

\vspace{-0.4cm}

\subsection{Minimum Penalty Algorithm}

\begin{algorithm} 
\caption{Minimum Penalty Algorithm (MPA)}  \label{algo_MPA}
\begin{algorithmic}[1] 
\small \STATEx \textbf{Input:} $\cal{F}$ \nonumber
\STATEx \textbf{Output:} $\cal{S}$, $t(\cal{S})$ \nonumber
\STATE $\cal{S}$ $\leftarrow$ $\emptyset$, $t(\cal{S})$ $\leftarrow$ 0,
\WHILE {$\textbf{F} \neq \emptyset$}
\STATE $k \leftarrow$ arg$\min_{i\in \cal{F}} \rho_i(t(\cal{S}))$,
\STATE $\cal{S}$ $\leftarrow$ $\cal{S}$ + $\lbrace k\rbrace$,
\STATE $\cal{F}$ $\leftarrow$ $\cal{F}$ - $\lbrace k \rbrace$,
\STATE $t(\cal{S})$ $\leftarrow$ $t(\cal{S})$ + $\rho_k(t(\cal{S}))$ + $t_k^{min}$,
\ENDWHILE
\end{algorithmic}
\end{algorithm}   

Minimum Penalty Algorithm (MPA) aims at minimizing the sum of the penalties of the users in a greedy manner based on Lemma \ref{lemma_equal}, as given in Algorithm \ref{algo_MPA}. Denote the schedule by $\cal{S}$, where the $i^{th}$ element of $\cal{S}$ is the user scheduled in the $i^{th}$ time slot. Let the schedule length be $t(\cal{S})$. Input of MPA algorithm is a set of energy harvesting users, denoted by $\cal{F}$. The algorithm starts by initializing the schedule $\cal{S}$ to an empty set and the schedule length $t(\cal{S})$ to $0$ (Line $1$). At each step of the algorithm, MPA picks the user with the minimum penalty among the unscheduled users (Line $3$). The current time slot is allocated to this minimum penalty user (Line $4$). Then, the scheduled user is discarded from set $\cal{F}$ (Line $5$) and the schedule length $t(\cal{S})$ is updated by adding the transmission time of the scheduled user (Line $6$). Algorithm terminates when all users in $\cal{F}$ are scheduled (Line $2$) and outputs the schedule $\cal{S}$ and the corresponding schedule length $t(\cal{S})$. The computational complexity of MPA is $\mathcal{O}(N^{2})$ for $N$ users since the algorithm picks one user at each iteration; i.e., $N$ iterations, and determines the minimum penalty user at each iteration, which has $\mathcal{O}(N)$ complexity. Fig. \ref{fig:pen3} illustrates MPA graphically for $2$ users.

\begin{figure}[t]
 \centering
\includegraphics[width= 0.65 \linewidth]{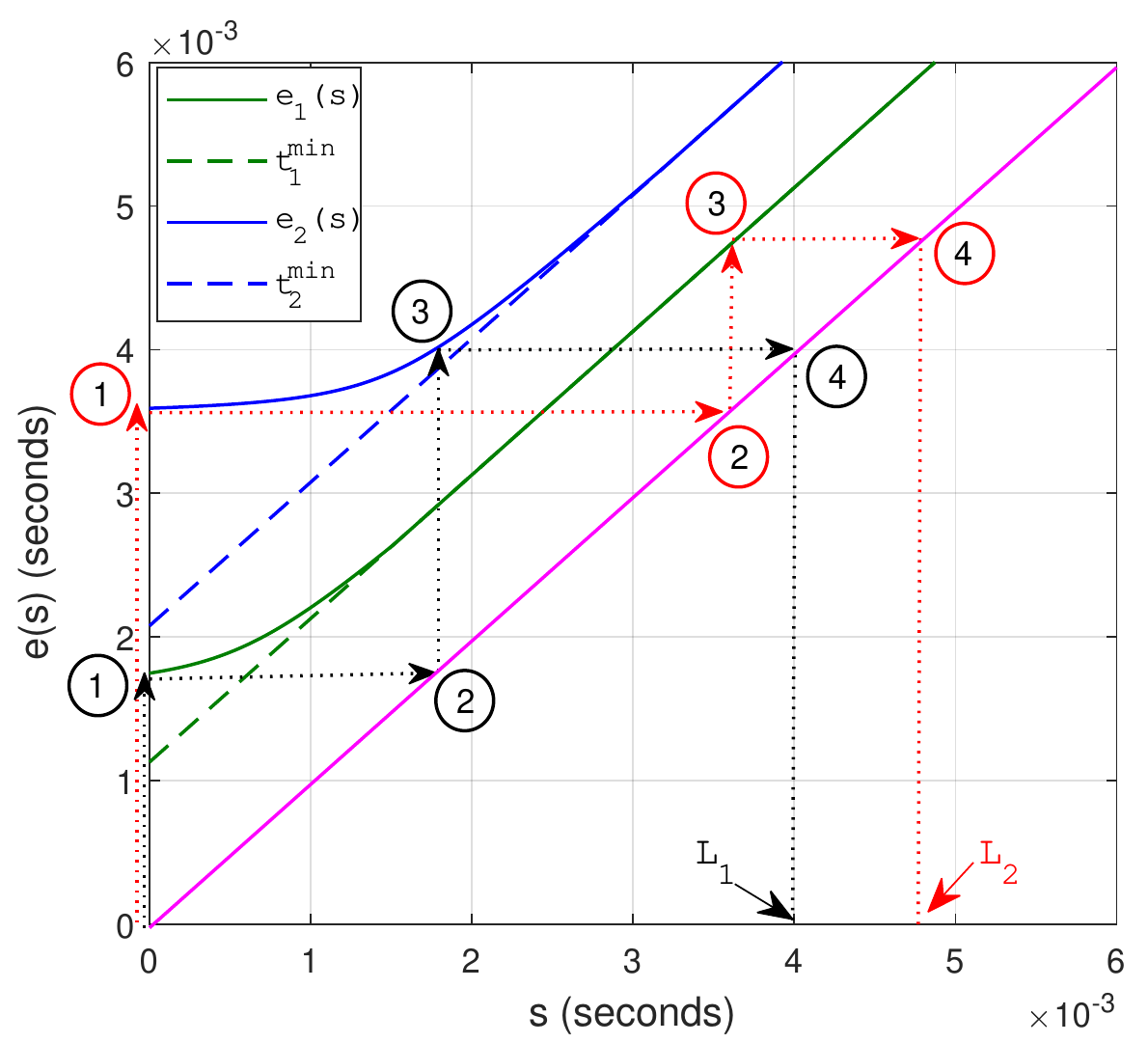}
\caption{Graphical illustration of MPA algorithm for $2$ users. Steps $1-4$ (dotted black line) depict the process of MPA. MPA first allocates user $1$ at $t=0$ since $\rho_1(0) < \rho_2(0)$. Then, after user $1$ completes its transmission at $t=e_1(0)$, user $2$ starts its transmission; i.e. $s_2 = e_1(0)$. The length of the resulting schedule is $L_1=e_2\left( e_1(0)\right) $. The alternative schedule, i.e., scheduling first user $2$ with a greater initial penalty and then user $1$ (dotted red line), would yield a greater schedule length $L_2>L_1$.}  \label{fig:pen3}
\end{figure}

\vspace*{-5mm}
\subsection{Maximum Transmit Power Algorithm}

\begin{algorithm}
\caption{Maximum Transmit Power Algorithm (MTPA) } \label{algo_MTPA}
\begin{algorithmic}[1]
\small \STATEx \textbf{Input:} $\cal{F}$
\STATEx \textbf{Output:} $\cal{S}$, $t(\cal{S})$
\STATE $\cal{S}$ $\leftarrow$ $\emptyset$, $t(\cal{S})$ $\leftarrow$ 0,
\WHILE {$\textbf{F} \neq \emptyset$}
\STATE determine optimal power $P_i$ for all $i\in \cal{F}$ at time $t(\cal{S})$,
\STATE $k \leftarrow$ arg$\max_{i\in \cal{F}} P_i$,
\STATE $\cal{S}$ $\leftarrow$ $\cal{S}$ + $\lbrace k\rbrace$,
\STATE $\cal{F}$ $\leftarrow$ $\cal{F}$ - $\lbrace k \rbrace$,
\STATE $t(\cal{S})$ $\leftarrow$ $t(\cal{S})$ + $\tau_k$,
\ENDWHILE
\end{algorithmic}
 \end{algorithm} 

Maximum Transmit Power Algorithm (MTPA) picks the user that can afford maximum feasible transmit power among all users, based on Corollary \ref{corollary_opt1}, where we show that allocating a user that can feasibly afford $P_{max}$ at any time instant is optimal. MTPA is given in Algorithm \ref{algo_MTPA}. Input of MTPA algorithm is a set of users denoted by $\cal{F}$ while the output is the schedule $\cal{S}$ and the corresponding schedule length $t(\cal{S})$. The algorithm starts by initializing $\cal{S}$ to an empty set and $t(\cal{S})$ to $0$ (Line $1$). At each step of the algorithm, optimal transmit power of each user $i$ is determined by Eq. (\ref{eq:opt_pow_1}) (Line $3$) and the user with maximum power is allocated to the current time slot with duration $\tau_i$ given by Eq. (\ref{eq:opt_pow_4}) (Lines $4-5$). Then, the scheduled user is discarded from set $\cal{F}$ (Line $6$) and the schedule length $t(\cal{S})$ is updated by adding the transmission time of the scheduled user (Line $7$). Algorithm terminates when all users in $\cal{F}$ are scheduled (Line $2$). The computational complexity of MTPA is $\mathcal{O}(N^{2})$ for $N$ users since the algorithm picks one user at each iteration; i.e., $N$ iterations, and determines the user with maximum transmit power at each iteration, which has $\mathcal{O}(N)$ complexity.

\vspace*{-5mm}
\subsection{Fast Pruning Algorithm} \label{sec:FPA}

In this section, we propose an exact and computationally-efficient enumeration algorithm incorporating efficient pruning mechanisms based on the analysis presented in Theorem \ref{thm_pen} and the following corollaries, as given in Algorithm \ref{algo_FPA}.

\begin{algorithm}
\caption{Fast Pruning Algorithm (FPA) }\label{algo_FPA}
\begin{algorithmic}[1]
\small \STATEx \textbf{Input:} $\cal{F}$
\STATEx \textbf{Output:} $\cal{S^*}$, $t^*$
\STATE $\cal{N}$ $\leftarrow$ $\lbrace \lbrace 1\rbrace ,  \lbrace 2\rbrace , ... ,   \lbrace \cal{F} \rbrace  \rbrace$, $t^*$ $\leftarrow$ $\infty$
\WHILE {$\cal{N}$ $\neq \emptyset$}
\STATE determine $s_{max}$ of $\cal{N}$,
\STATE ${\cal{N}}_{s_{max}}$ $\leftarrow$ the set of nodes in $\cal{N}$ with degree $s_{max}$,
\IF {$s_{max} = |\cal{F}|$}
\STATE $n_{max}$ $\leftarrow$ the node with size $s_{max}$,
\IF {$t(n_{max}) < t^*$}
\STATE $t^* \leftarrow t(n_{max})$,
\STATE $\cal{S^*}$ $\leftarrow$ $n_{max}$,
\ENDIF
\STATE discard node $n_{max}$ from $\cal{N}$,
\STATE continue,
\ENDIF
\STATE $n_{min}$ $\leftarrow$ arg$\min_{n\in {\cal{N}}_{s_{max}}} \rho(n)$,
\IF {$\rho(n_{min})=0$}
\STATE prune all nodes in $\cal{N}$ with size $s_{max}$ and same ascendant node with $n_{min}$,
\ENDIF
\IF {$t(n_{min}) \geq t^*$}
\STATE prune node $n_{min}$,
\ELSE
\STATE generate set of children nodes ${\cal{C}}_{n_{min}}$ of node $n_{min}$,
\STATE $\cal{N}$ $\leftarrow$ $\cal{N}$ + ${\cal{C}}_{n_{min}}$,
\ENDIF
\STATE discard node $n_{min}$ from $\cal{N}$,
\ENDWHILE
\end{algorithmic}
\end{algorithm}

Consider a tree model with $|\cal{F}|$ levels, where $\cal{F}$ is the set of users to be scheduled. Each level $i$ of the tree specifies the user allocated in the $i^{th}$ time slot of the schedule. A branch of the tree, consisting of one user from each level of the tree, corresponds to one feasible schedule for $\cal{F}$. A node $n$ of the tree at the $i^{th}$ level specifies the set of $i$ users allocated in the first $i$ slots of the schedule. Size $s(n)$ of a node $n$ is defined as the number of scheduled users specified by $n$; i.e., a node at the $i^{th}$ level of the tree has size $i$. For instance, a node $n=\left\lbrace 2, 3 \right\rbrace $ of the tree at the $2^{nd}$ level specifies that users $2$ and $3$ are allocated in the first and second slots of the schedule, respectively, and has size $s(n)=2$. Note that maximum size for a node is $|\cal{F}|$ corresponding to a branch; i.e., a feasible schedule consisting of all users. We further define the penalty $\rho(n)$ and the transmission length $t(n)$ of a node $n$ as the penalty of the last scheduled user by $n$ and the sum of the transmission times of the users scheduled by $n$, respectively.

The aim of the Fast Pruning Algorithm (FPA) is to determine the optimal schedule without generating all possible branches of the tree; i.e. without generating all feasible schedules. Two pruning mechanisms are employed to decrease the search space for the feasible schedules in this regard. Unless a node of the tree is pruned out by the algorithm via these mechanisms, it is branched into children nodes, each corresponding to the allocation of a new user to the already scheduled users by the ascendant node. First pruning mechanism is based on Corollary \ref{corollary_opt1}. If the penalty of a node is $0$, all children nodes of its ascendant are pruned out since the other nodes originating from the same ascendant cannot yield a better schedule than that particular node. Second one is pruning out the nodes which cannot end up in a minimum length schedule. If the transmission length of a node is greater than or equal to the current minimum schedule length, that particular node is pruned out since adding new users to a node will only increase its transmission length.

Input of the FPA algorithm is a set of users denoted by $\cal{F}$, while the output is the schedule $\cal{S}^{*}$ and the corresponding schedule length $t^{*}$. FPA algorithm keeps track of the set of nodes $\cal{N}$ of the tree that are not evaluated so far. The algorithm starts by initializing $\cal{N}$ to set $\lbrace \lbrace 1\rbrace ,  \lbrace 2\rbrace , ... ,   \lbrace \cal{F} \rbrace  \rbrace$, containing all nodes in the $1^{st}$ level of the tree, each corresponding to the allocation of one user in $\cal{F}$ to the $1^{st}$ time slot and $t^{*}$ to $\infty$ (Line $1$). Let $s_{max}$ be the largest size of a node in $\cal{N}$ and ${\cal{N}}_{s_{max}}$ be the set of nodes in $\cal{N}$ with size $s_{max}$. FPA determines $s_{max}$ and ${\cal{N}}_{s_{max}}$ at each iteration (Lines $3-4$). Unless $\cal{N}$ includes a node with size $|\cal{F}|$; i.e., a branch, FPA picks the node with the minimum penalty in ${\cal{N}}_{s_{max}}$, denoted by $n_{min}$ (Line $14$). If the penalty of node $n_{min}$ is $0$ (Line $15$), then all nodes with the same size $s_{max}$; i.e., having same ascendant node with $n_{min}$, are pruned out (Line $16$). If the transmission length of node $n_{min}$ is greater than or equal to the current minimum schedule length $t^{*}$ (Line $18$), then $n_{min}$ is also pruned out (Line $19$). Otherwise, $n_{min}$ is branched into its children nodes, denoted by set ${\cal{C}}_{n_{min}}$ (Line $21$), which is added to set $\cal{N}$ (Line $22$). Finally, $n_{min}$ is discarded from $\cal{N}$ since its evaluation is completed (Line $24$). If, at a particular iteration, $\cal{N}$ includes a node with maximum size $|\cal{F}|$ (Line $5$), then the node $n_{max}$ with degree $s_{max}$ (Line $6$) is a branch and specifies a feasible schedule. Therefore, algorithm evaluates whether it outperforms the current best feasible schedule $S^*$, and updates $S^*$ and the corresponding current minimum schedule length $t^*$ if the transmission length of this node is less than $t^{*}$ (Lines $7-10$). Then, $n_{max}$ is discarded from $\cal{N}$ (Line $11$) and the algorithm continues with the next iteration (Line $12$). FPA terminates when all nodes in $\cal{N}$ are evaluated; i.e., $\cal{N}=\emptyset$ (Line $2$). Fig. \ref{fig:FPA} illustrates FPA algorithm for $4$ users, through an example.

\begin{figure}[t]
 \centering
\includegraphics[width= 0.65 \linewidth]{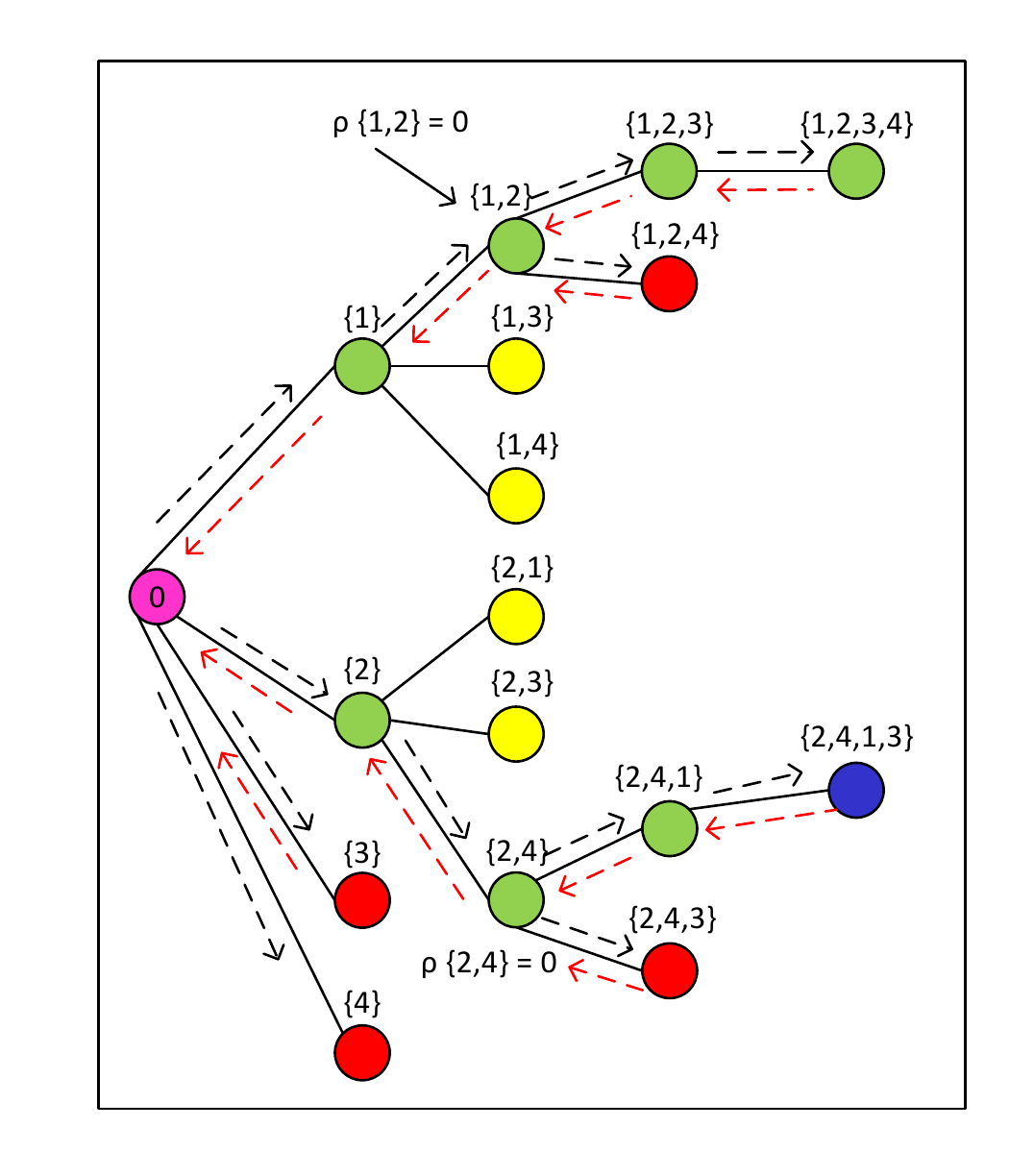}
\caption{Graphical illustration of FPA algorithm for $4$ users. Algorithm starts by generating the $1^{st}$ level nodes $\lbrace 1 \rbrace$, $\lbrace 2 \rbrace$, $\lbrace 3 \rbrace$, and $\lbrace 4 \rbrace$. The green circles represent the nodes evaluated and not pruned by FPA. The yellow circles represent the pruned nodes having the same ascendant with a zero penalty node while the red circles represent the pruned nodes that cannot yield a better schedule compared to the current best feasible schedule determined by FPA. Blue circle denotes the optimal schedule yielded by FPA after all nodes are either evaluated or pruned out. The nodes are evaluated by FPA in the following order: $\lbrace 1 \rbrace\rightarrow \lbrace 1,2 \rbrace\rightarrow \lbrace 1,2,3 \rbrace\rightarrow \lbrace 1,2,3,4 \rbrace\rightarrow \lbrace 1,2,4 \rbrace\rightarrow \lbrace 2 \rbrace\rightarrow \lbrace 2,4 \rbrace\rightarrow \lbrace 2,4,1 \rbrace\rightarrow \lbrace 2,4,1,3 \rbrace\rightarrow \lbrace 2,4,3 \rbrace\rightarrow \lbrace 3 \rbrace\rightarrow \lbrace 4 \rbrace$.}  \label{fig:FPA}
\end{figure}

\begin{theorem} \label{thm:FPA}
FPA algorithm determines an optimal solution for $\cal{MLSP}$.
\end{theorem}

\begin{proof}
We will first show that FPA generates all possible schedules for the set of users $\cal{F}$ unless a node of the tree is pruned out. Consider any node $n=\lbrace u_1, u_2,..., u_{|\cal{F}|}\rbrace$ with size $s(n)= |\cal{F}|$ corresponding to a feasible schedule for the set of users $\cal{F}$, where $u_i$ denotes the user allocated in the $i^{th}$ time slot of the schedule represented by node $n$. The ascendant of node $n$, say $n^p$, is the node $n^p=\lbrace u_1, u_2,..., u_{{|\cal{F}|}-1}\rbrace$ with size $s(n^p)={|\cal{F}|}-1$ and corresponds to the allocation of first ${|\cal{F}|}-1$ users scheduled by node $n$. Then, by the same logic, one can iteratively determine the ascendants of $n^p$ up to the node $n^i=\lbrace u_1\rbrace$ with size $s(n^i)=1$ corresponding to the allocation of only the first user scheduled by node $n$. Node $n^i$ is one of the nodes in the initially specified set $\cal{N}$ by FPA algorithm (Line $1$). This guarantees the generation of the schedule represented by node $n$ by FPA unless one of the ascendants $\lbrace n^i,...,n^p\rbrace$ of node $n$ is pruned out. Furthermore, since the branch of the tree originating from node $n^i$ and consisting of the successive ascendants of node $n$ up to $n^i$ is uniquely determined, the schedule represented by node $n$ is generated only once. On the other hand, any node $n$ pruned out by FPA cannot yield an optimal solution since it is pruned either due to the existence of a better feasible schedule determined by FPA so far (Lines $18-20$) or because one node having the same ascendant with $n$ has $0$ penalty and hence yields a better schedule (Lines $15-17$). Thus, the optimal schedule will be one of the nodes generated by FPA. This completes the proof.
\end{proof}

Note that the proposed FPA algorithm is not a conventional branch and bound method. The basic idea of a standard branch and bound method is to partition the non-convex feasible region into convex sets and finding the upper and lower bounds on the global optimum iteratively based on the solution of the convex relaxations of the integer programming problems. The method terminates when the upper and lower bounds reach a sufficiently close neighborhood of each other. On the other hand, the FPA algorithm is an enumeration algorithm that evaluates integral feasible solutions without performing any relaxation. Nevertheless, to improve the speed of convergence to the optimal solution, the FPA algorithm incorporates two pruning mechanisms, similar to the branching and bounding routines used in the branch and bound method. The first pruning mechanism comparing the solution of a feasible schedule to the upper bound on the optimal solution is similar to the one used in standard branch and bound. However, the other pruning mechanism exploiting the introduced penalty function is particularly designed for $\cal{MLSP}$. Hence, the FPA algorithm can be considered as a smart enumeration algorithm specifically designed for the investigated problem, exploiting the introduced novel penalty function for smart and fast pruning.

The computational complexity of FPA is $\mathcal{O}(N!)$ since it enumerates all possible schedules for $N$ users in the worst case as stated in the proof of Theorem \ref{thm:FPA}. In this respect, it is identical to the worst case complexity of standard branch and bound. However, as illustrated in Section \ref{sec:simulation}, FPA achieves a much better runtime performance in a practical setup than its worst case complexity due to the incorporated smart pruning mechanisms.

\vspace{-0.5cm}
\section{Performance Evaluation} \label{sec:simulation}
 
The goal of this section is to evaluate the performance of the proposed algorithms in comparison to the optimal solution and previously proposed algorithms. The previously proposed algorithm in \cite{harvest_30}, denoted by PCA, aims at minimizing the schedule length in a full-duplex system for a given transmission order of the users, without considering scheduling and maximum transmit power constraint. The algorithm determining the optimal time and power allocation for a predetermined transmission order of users based on Theorem \ref{thm:opt_power} is included for a fair comparison to PCA, denoted by OTPA. In order to illustrate the computational performance of the FPA algorithm, we also incorporate a brute-force algorithm, denoted by BFA, which evaluates the length of all possible transmission schedules and determines the one with minimum length using PCA to determine the length of each transmission schedule. Note that since FPA is proven to be optimal, the figures illustrating the schedule length performance, Figures $\ref{fig:pmax}-\ref{fig:beta}$, exclude BFA.

Simulation results are obtained by averaging $1000$ independent random network realizations. The users are uniformly distributed in a circle with radius of $10$m. The attenuation of the links considering large-scale statistics are determined using the path loss model given by 
$PL(d)=PL(d_0)+10\alpha log_{10}\bigg(\frac{d}{d_0}\bigg)+\emph{Z}$,
where $PL(d)$ is the path loss at distance $d$, $d_0$ is the reference distance, $\alpha$ is the path loss exponent, and $Z$ is a zero-mean Gaussian random variable with standard deviation $\sigma$. The small-scale fading has been modeled by using Rayleigh fading with scale parameter  $\Omega$ set to mean power level obtained from the large-scale path loss model. The parameters used in the simulations are $\eta_i=1$ for $i \in \{1,\dots, N\}$; $D_i=100$ bits for $i \in \{1,\dots, N\}$; $W= 1$ MHz; $d_0=1$ m; $PL(d_0)=30$ dB; $\alpha=2.76$, $\sigma=4$ \cite{harvest_07, harvest_04}, \cite{harvest_50}. The self interference coefficient $\beta$ is $-70$ dBm, the initial battery level of the users are $10^{-9}$ J, $P_{max}=1mW$ and $P_h=1W$ for the simulations, unless otherwise stated.

\begin{figure}[t] 
 \centering
\includegraphics[width= 0.65 \linewidth]{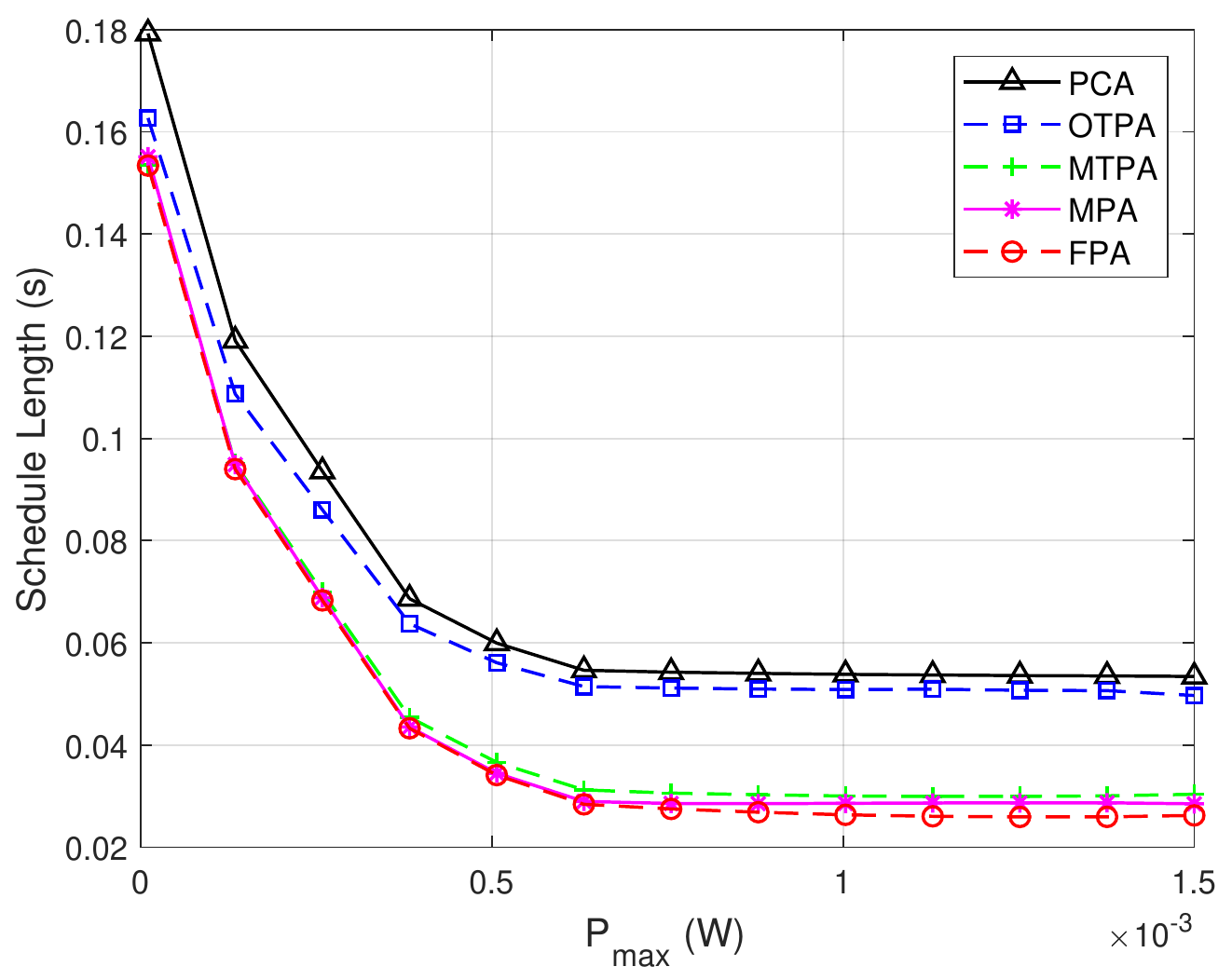}
\caption{Schedule length of the algorithms for different maximum transmit powers in a network of $10$ users and $P_h=30W$}. \label{fig:pmax}
\end{figure}

Fig. \ref{fig:pmax} illustrates the schedule length for different $P_{max}$ values in a network of $10$ users and $P_h=30W$. The proposed scheduling algorithms; i.e., MPA, MTPA and FPA, perform significantly better than the algorithms designed for a given scheduling order; i.e., PCA and OTPA, approximately yielding around $50\%$ performance improvement for practical $P_{max}$ values. The performance gain of scheduling increases as $P_{max}$ increases up to $0.6W$. This is due to the fact that as $P_{max}$ value increases, the first time instant at which a user has zero penalty increases necessitating a proper ordering of the users. Note that the scheduling performance of the algorithms saturates around $P_{max}=0.6W$. This is due to the fact that most users are bounded by the energy causality constraint around this level and cannot afford higher transmit powers even if $P_{max}$ further increases. Moreover, MPA and MTPA perform very close to optimal and show robustness against increasing $P_{max}$. In addition, MPA outperforms MTPA by approximately $5\%$, which is considerable given that both algorithms perform very close to optimal. This performance improvement of MPA over MTPA is due to the fact that the former directly aims at minimizing the completion time of the user transmissions through penalty minimization at each user allocation, while the latter considers maximizing the transmit powers, which may not correspond to picking users with minimum penalty since users have different energy harvesting and consumption characteristics. Note that the penalty function inherently takes data requirements, channel conditions and available energy into account, while the transmit power of a user is just a function of the energy available for a user at a time.

\begin{figure}[t]
 \centering
\includegraphics[width= 0.65 \linewidth]{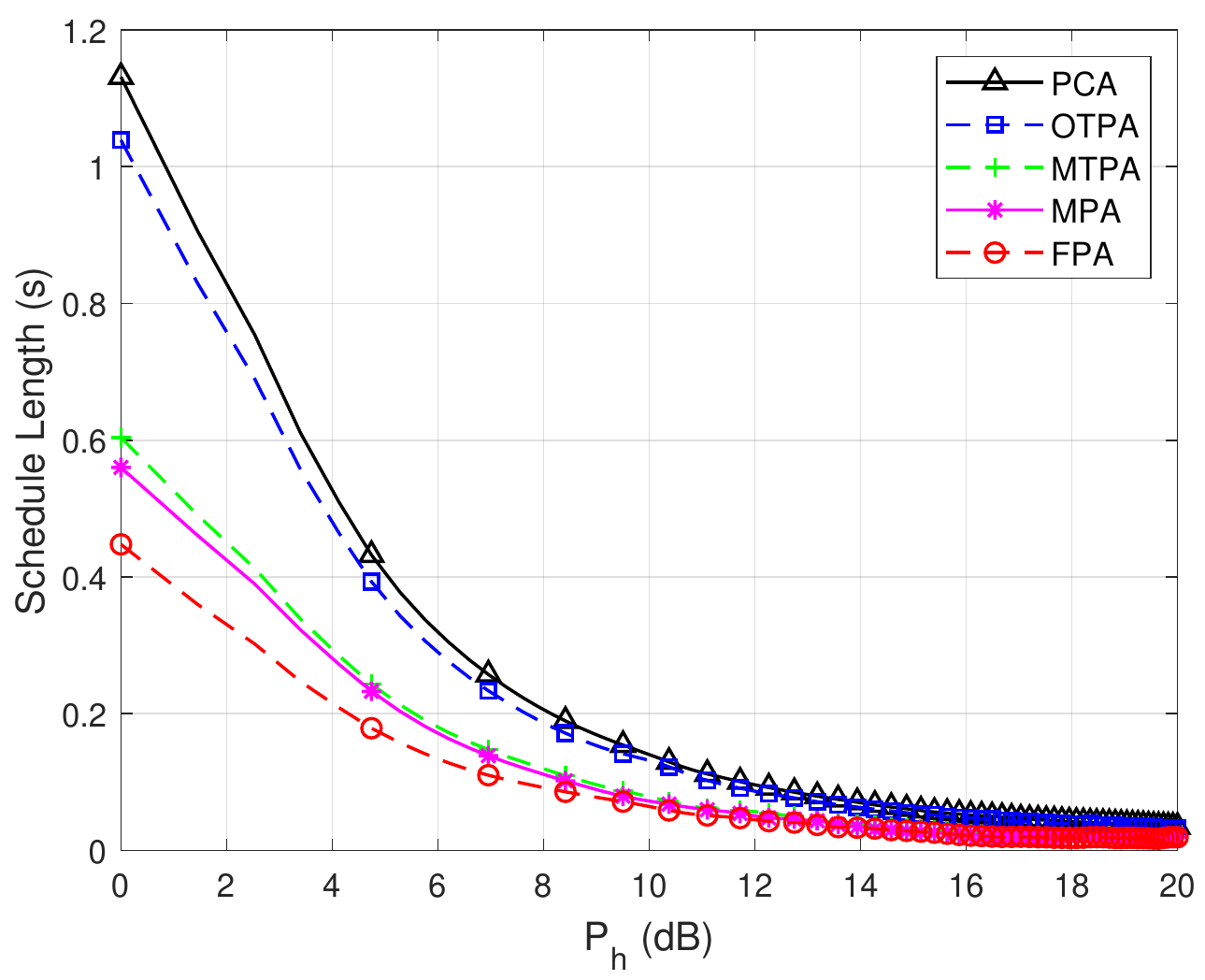}
\caption{Schedule length of the algorithms for different HAP transmit powers in a network of $10$ users.} \label{fig:ph}
\end{figure}

Fig. \ref{fig:ph} illustrates the schedule length for different HAP transmit power values in a network of $10$ users. The schedule length decreases with the increasing transmit power since higher HAP power allows users to harvest more energy and complete their transmission in shorter time as long as they can transmit with higher transmit powers. The proposed polynomial-time algorithms MPA and MTPA employing scheduling complete the transmission of all users in around $40\%$ and $50\%$ less time with respect to OTPA and PCA, respectively, for small $P_h$ values. Furthermore, MPA and MTPA achieve an approximation ratio of around $1.20$ and $1.25$, respectively, where approximation ratio is defined as the ratio of the scheduling length of the heuristic algorithms to the optimal schedule length yielded by FPA. On the other hand, for large $P_h$ values, around $20$ dB, the performance of the algorithms are very close to each other, since users can reach maximum transmit power $P_{max}$ values faster, which removes the necessity for scheduling: Each user $i$ transmits at high data rate and can complete its transmission in a time slot with length close to $\tau_i^{min}$. Similar to Fig. \ref{fig:pmax}, MPA and MTPA perform very close to optimal, both achieving an approximation ratio below $1.10$ for practical $P_h$ values, indicating that the algorithms can adapt their schedules based on the energy harvested by each user efficiently.

\begin{figure}[t] 
 \centering
\includegraphics[width= 0.65 \linewidth]{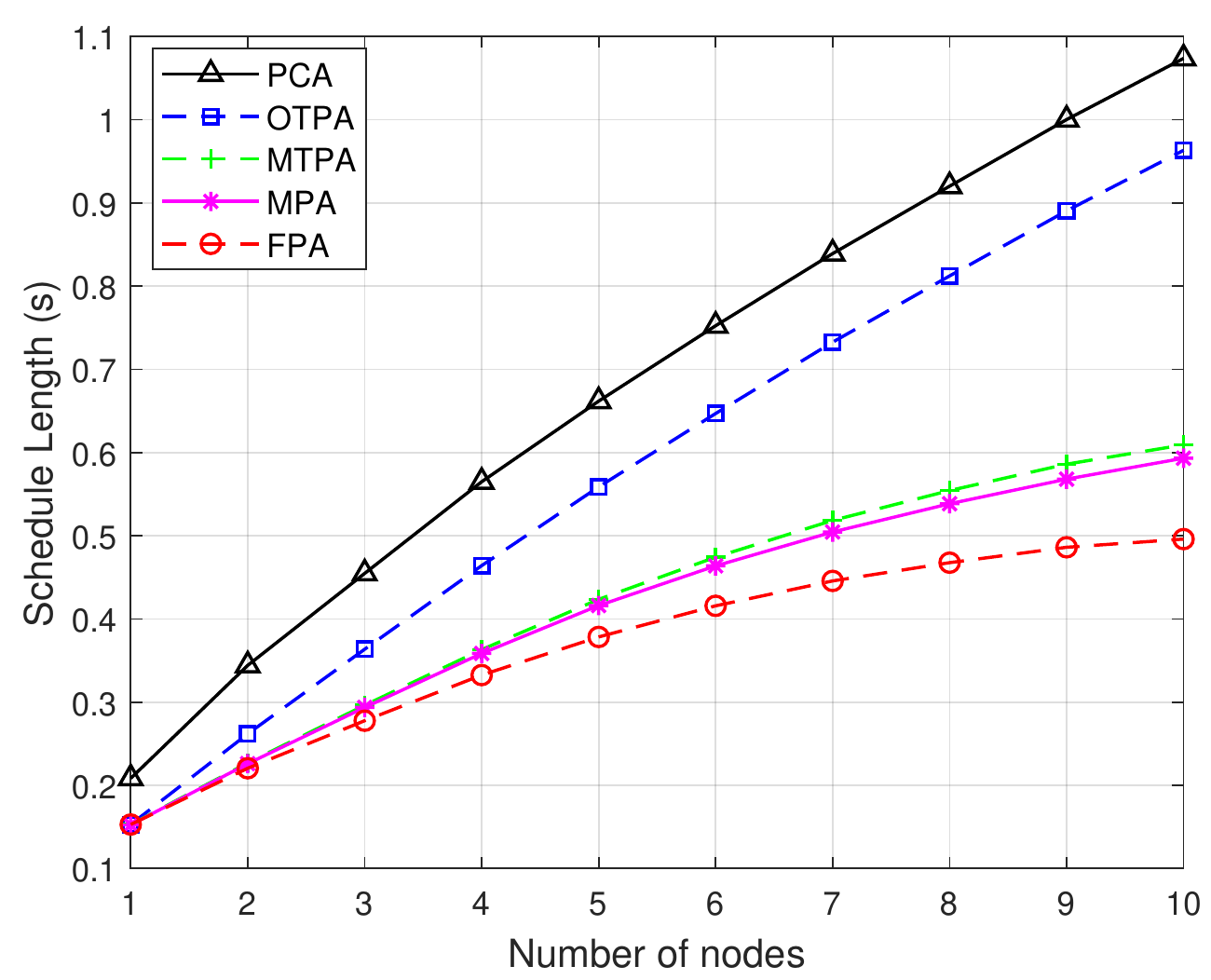}
\caption{Schedule length of the algorithms for different number of users.}\label{fig:size}
\end{figure}

Fig. \ref{fig:size} illustrates the effect of the network size on the performance of the proposed algorithms. For OTPA and PCA in which no scheduling is performed, the addition of each user increases the schedule length by almost a constant amount since increase in the schedule length is just caused by the time slot length of that particular user. On the other hand, for MTPA, MPA and FPA, each new user yields a diminishing increase in the schedule length since the transmission order may change by the addition of a user. For instance, a new user with good channel conditions, which can initially transmit with $P_{max}$, will be scheduled first and this will allow the other users to harvest more energy to be able to complete their transmissions in shorter time durations. Note that, as discussed in the previous sections, delaying the transmission of a user decreases its transmission time if its transmission power increases. Then, the increase in the schedule length will be less than the time slot length of the new user by the sum of penalty reductions of the other users. MPA and MTPA perform very close to optimal as the number of users increases, achieving approximation ratios of $1.16$ and $1.20$ at maximum, respectively, for $10$ users caused by the exponential nature of the problem complexity. Note that for $10$ users, there are $10!$ possible schedules among which MPA and MTPA determine only one in polynomial-time. However, their performance is still very robust to the network size since having larger number of users increases the probability of having at least one user with zero penalty or $P_{max}$ transmit power to be picked by MPA and MTPA at a time, respectively. 

\begin{figure}[t] 
 \centering
\includegraphics[width= 0.65 \linewidth]{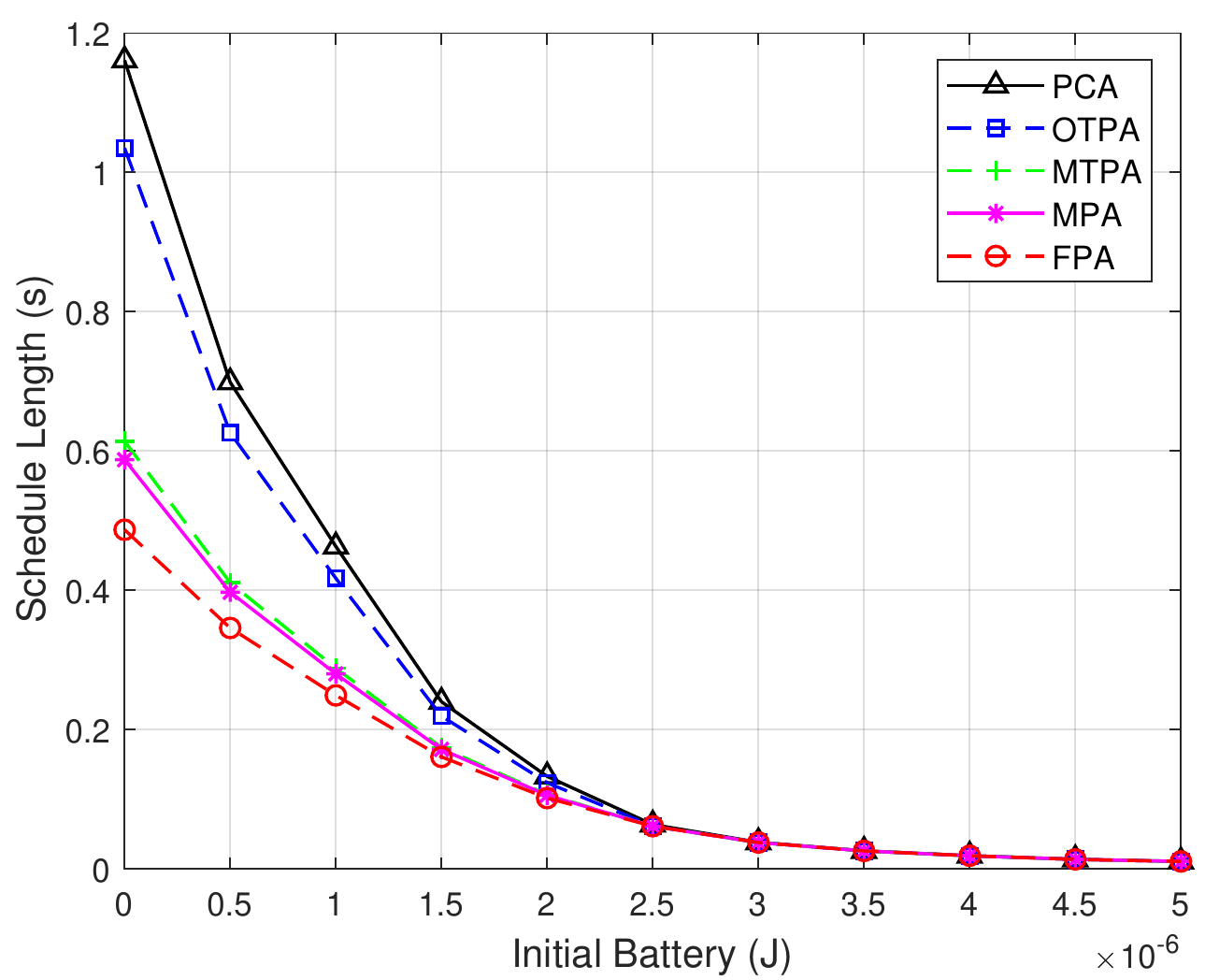}
\caption{Schedule length of the algorithms for different initial battery levels in a network of $10$ users.}\label{fig:battery}
\end{figure}

Fig. \ref{fig:battery} illustrates the schedule length for different initial battery levels in a network of $10$ users, assuming that each user has the same initial battery level. For low initial battery levels, the total transmission time of all users is approximately $45\%$ lower for the proposed heuristic scheduling algorithms MPA and MTPA than the pre-determined transmission order. However, as the initial battery level increases, the impact of the scheduling decreases due to the fact that higher initial battery level can allow more users to transmit information by using $P_{max}$. When the initial battery levels of all users are high enough, all the users can afford transmission at $P_{max}$, eliminating the effect of scheduling. The slight improvement in the performance of the OTPA over PCA is due to the energy harvesting capability of a user during its own transmission.

\begin{figure}[t] 
 \centering
\includegraphics[width= 0.65 \linewidth]{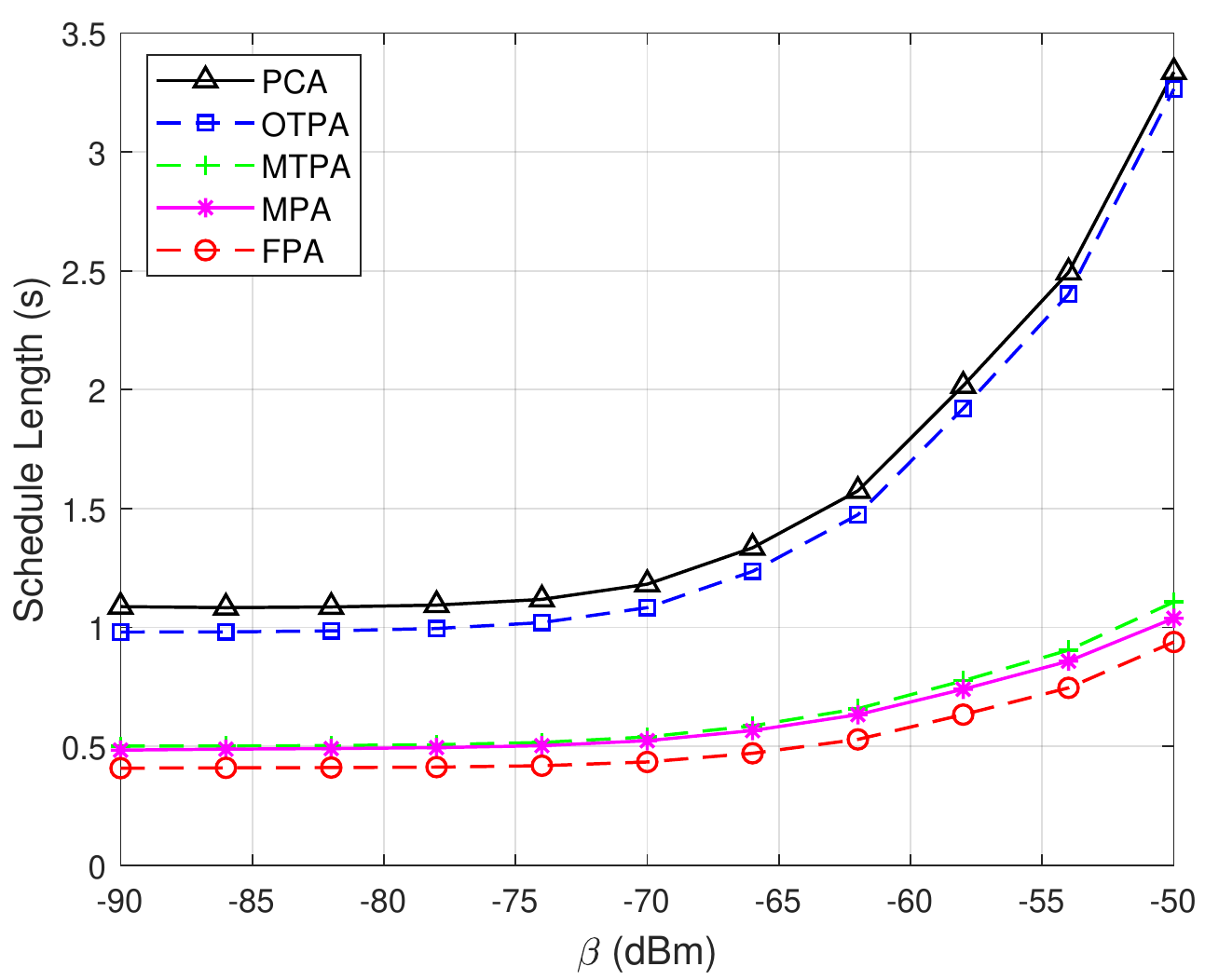}
\caption{Schedule length of the algorithms for different interference levels in a network of $10$ users.}\label{fig:beta}
\end{figure}
Fig. \ref{fig:beta} illustrates the schedule length for different values of the self interference coefficient $\beta$. For lower $\beta$ values, i.e. $-90$dBm, the impact of self interference is very low, resulting in a lower schedule length. However, as the self-interference level increases, the rate of increase in the schedule length increases, mainly due to the dominating effect of the power of self-interference over noise. Moreover, for all self-interference levels, the proposed scheduling algorithms perform very close to optimal and better than the predetermined transmission schemes by approximately $50\%$ and $60\%$ for low and high $\beta$ values, respectively. The difference between the proposed algorithms and fixed transmission order increases as the self-interference level increases, mainly due to decreasing transmission rates and correspondingly increasing penalty values, which in turn necessitates a proper ordering of the user transmissions.

\begin{figure}[t] 
 \centering
\includegraphics[width= 0.65 \linewidth]{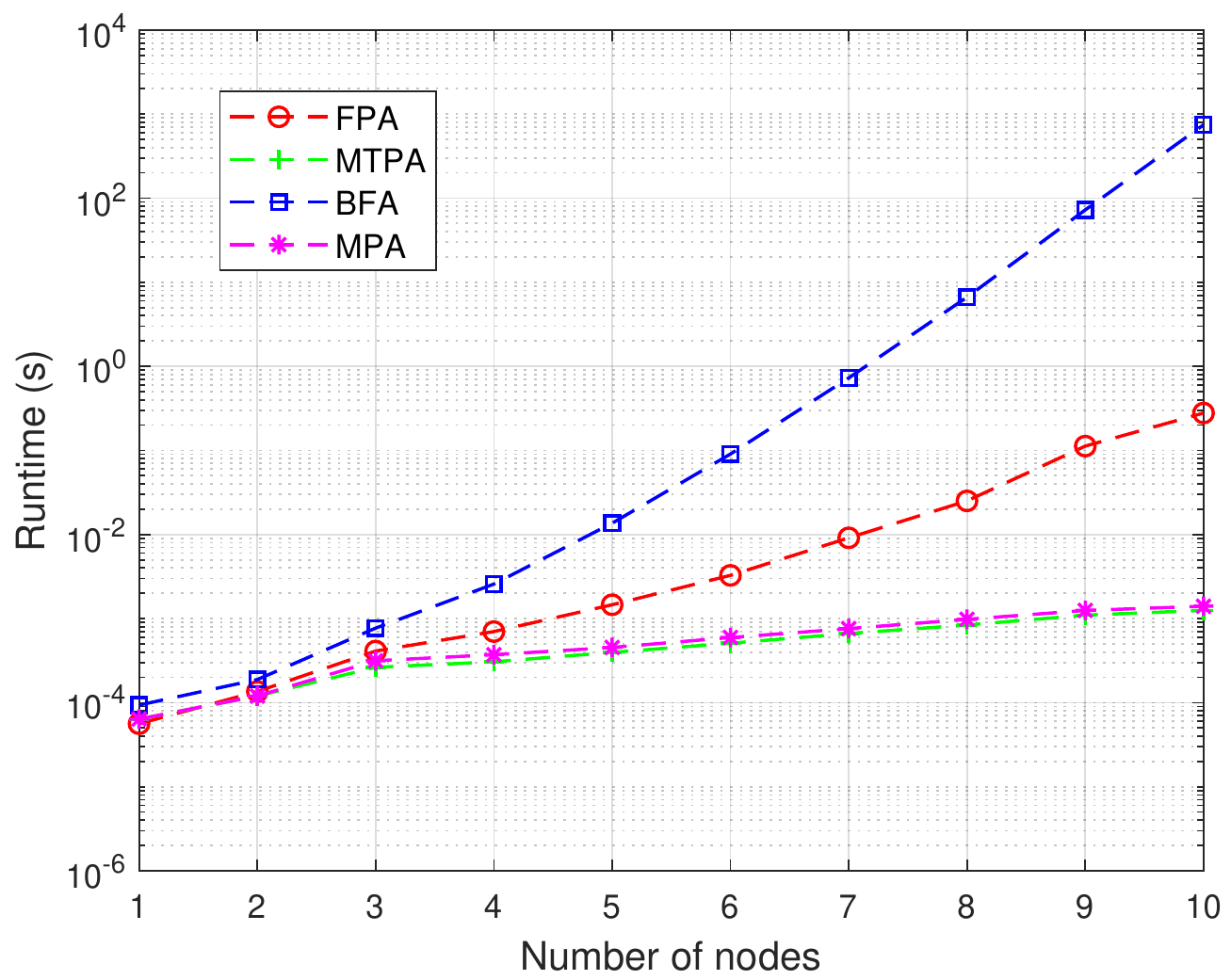}
\caption{Runtime of the algorithms for different number of users.} \label{fig:runtime}
\end{figure}

Fig. \ref{fig:runtime} shows the average runtime of the proposed algorithms for increasing number of users. The runtime of BFA increases exponentially with increasing number of users, therefore, has a large computational burden. On the other hand, the proposed optimal algorithm FPA decreases the runtime significantly by reducing the search space for the optimum schedule via smart pruning mechanisms. For a network of $10$ users, FPA achieves a runtime around one thousandth of the runtime of BFA. Furthermore, the runtime of the proposed polynomial-time algorithms MTPA and MPA increases almost linearly, as expected. Evaluating Fig. \ref{fig:size} and Fig. \ref{fig:runtime} together, we can observe that MTPA and MPA are scalable algorithms that can achieve close to optimal solutions in reasonable runtimes even for large network sizes, for which use of an exponential-time algorithm would be intractable.

\vspace{-0.5cm}
\section{Conclusion} \label{sec:conclusion}

In this paper, we have considered a WPCN where multiple users can harvest energy from and communicate data to a hybrid access point that can supply RF energy in full duplex manner. We have investigated the minimum length scheduling problem to determine the optimal power control, time allocation and transmission schedule subject to data, energy causality and maximum transmit power constraints. We have formulated the problem as a MINLP problem, which is generally difficult to solve for a global optimum, and conjectured that the problem is NP-hard. We have provided a solution strategy in which the power control and time allocation, and the scheduling problems are decomposed. For the power control and time allocation problem, we have proposed optimal closed-form solution. For the scheduling, we have introduced the penalty function through which we have analyzed the characteristics of the optimal solution. We have proposed two polynomial time heuristic algorithms and an exact fast enumeration algorithm based on the optimality conditions. Through simulations, we have illustrated that the heuristic algorithms perform very close-to-optimal outperforming the conventional schemes significantly. 

\bibliographystyle{ieeetr}
\bibliography{bib_shahid}
\begin{IEEEbiography}[{\includegraphics[width=1in,height=1.25in,clip,keepaspectratio]{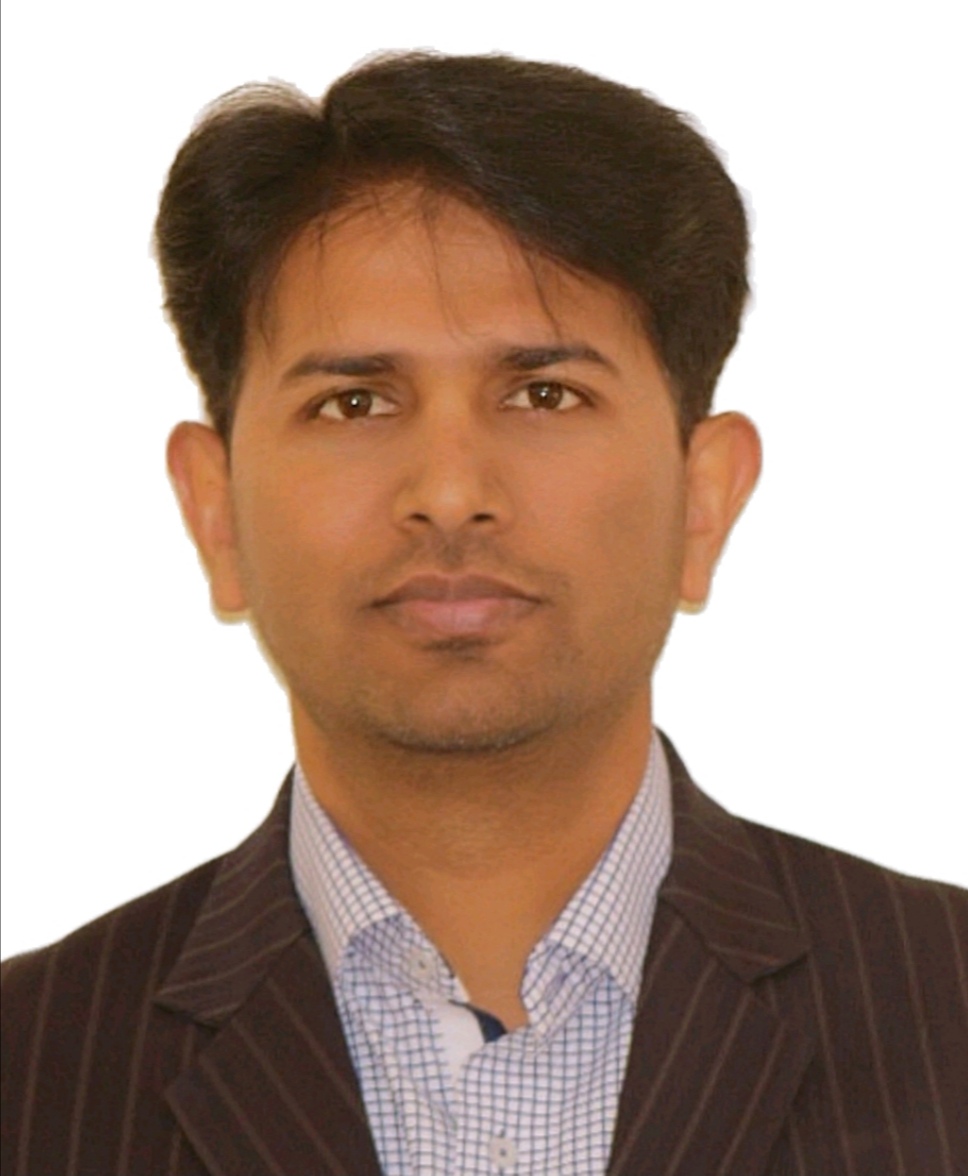}}]{Muhammad Shahid Iqbal}
 received the BS degree in telecom engineering from Government College University, Faisalabad, in 2009 and the MS degree in electrical engineering from National University of Sciences and Technology (NUST), Islamabad, Pakistan in 2013. Since 2016 he is pursuing his Ph.D. degree in electrical and electronics engineering from Koc University, Istanbul. He has been a lecturer with the Department of Electrical Engineering, The University of Lahore, from 2013 to 2014 and with the University of Gujrat from 2014 to 2016. His research interests include wireless communication design, wireless powered communication networks, and sensor networks. He received the higher education commission (HEC) Pakistan Fellowship in 2016 for full time Ph.D studies.
\end{IEEEbiography}
\vskip 0pt plus -1fil
\begin{IEEEbiography}[{\includegraphics[width=1in,height=1.25in,clip,keepaspectratio]{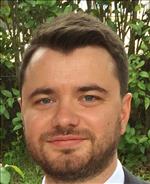}}]{Yalcin Sadi}
 received the B.S., M.S., and Ph.D. degrees in electrical and electronics engineering from Koc University, Istanbul, in 2010, 2012, and 2015, respectively. Since 2016, he has been an Assistant Professor with the Department of Electrical and Electronics Engineering, Kadir Has University, Istanbul. His research interests include wireless communication design and networking for cellular communications, machine-to-machine communications, and sensor networks. He received The Scientific and Technological Research Council of Turkey (TUBITAK) Career Development Program Grant in 2019, the TUBITAK Graduate Fellowship in 2010, and the Koc University Full Merit Scholarship in 2005. 
\end{IEEEbiography}
\vskip 0pt plus -1fil
\begin{IEEEbiography}[{\includegraphics[width=1in,height=1.25in,clip,keepaspectratio]{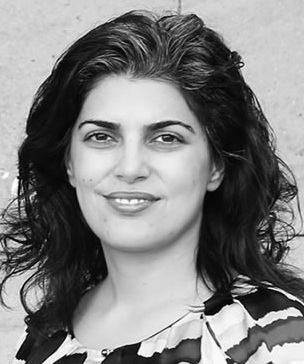}}]{Sinem Coleri}
 received the BS degree in electrical and electronics engineering from Bilkent University in 2000, the M.S. and Ph.D. degrees in electrical engineering and computer sciences from University of California Berkeley in 2002 and 2005. She worked as a research scientist in Wireless Sensor Networks Berkeley Lab under sponsorship of Pirelli and Telecom Italia from 2006 to 2009. Since September 2009, she has been a faculty member in the department of Electrical and Electronics Engineering at Koc University, where she is currently Professor. Her research interests are in wireless communications and networking with applications in cyber-physical systems, machine-to-machine communication, sensor networks and intelligent transportation systems. 
She received College of Engineering Outstanding Faculty Award at Koc University and IEEE Communications Letters Exemplary Editor Award as Area Editor in 2019, Outstanding Achievement Award by Higher Education Council and Academician of the Year Award by ANTIKAD (Antalya Businesswoman Association) in 2018, IEEE Communications Letters Exemplary Editor Award and METU- Prof. Dr. Mustafa Parlar Foundation Research Encouragement Award in 2017, IEEE Communications Letters Exemplary Editor Award and Science Heroes Association - Scientist of the Year Award in 2016, Turkish Academy of Sciences Distinguished Young Scientist (TUBA-GEBIP) and TAF (Turkish Academic Fellowship) Network - Outstanding Young Scientist Awards in 2015, Science Academy Young Scientist (BAGEP) Award in 2014, Turk Telekom Collaborative Research Award in 2011 and 2012, Marie Curie Reintegration Grant in 2010, Regents Fellowship from University of California Berkeley in 2000 and Bilkent University Full Scholarship from Bilkent University in 1995. She has been Area Editor of IEEE Communications Letters and IEEE Open Journal of the Communications Society since 2019, Editor of IEEE Transactions on Communications since 2017 and Editor of IEEE Transactions on Vehicular Technology since 2016. 
\end{IEEEbiography}
\end{document}